\documentclass[jair,twoside,11pt,theapa]{article}
\usepackage{jair, theapa, rawfonts}
\usepackage{amssymb}
\usepackage{amsmath}
\usepackage{amsthm}
\usepackage{mathtools}
\usepackage{booktabs}
\usepackage{multirow}
\usepackage{caption}
\usepackage[labelformat=simple]{subcaption}
\newtheorem{definition}{Definition}
\newtheorem{theorem}{Theorem}
\newtheorem{problem}{Problem}
\newtheorem{subproblem}{Subproblem}

\begin{document}

\title{Scalable Synthesis of Formally Verified Neural Value Function for Hamilton-Jacobi Reachability Analysis}

\author{\name Yujie Yang \email yangyj21@mails.tsinghua.edu.cn \\        
        \addr School of Vehicle and Mobility, Tsinghua University, \\
        Beijing, 100084, China \\
        \addr Robotics Institute, Carnegie Mellon University, \\
        Pittsburgh, PA, 15289, USA \\
        \AND
        \name Hanjiang Hu \email hanjianh@andrew.cmu.edu \\
        \name Tianhao Wei \email twei2@andrew.cmu.edu \\
        \addr Robotics Institute, Carnegie Mellon University, \\
        Pittsburgh, PA, 15289, USA \\
        \AND
        \name Shengbo Eben Li \email lishbo@tsinghua.edu.cn \\
        \addr School of Vehicle and Mobility \& College of AI, Tsinghua University, \\
        Beijing, 100084, China \\
        \AND
        \name Changliu Liu \email cliu6@andrew.cmu.edu \\
        \addr Robotics Institute, Carnegie Mellon University, \\
        Pittsburgh, PA, 15289, USA
}


\maketitle

\begin{abstract}
Hamilton-Jacobi (HJ) reachability analysis provides a formal method for guaranteeing safety in constrained control problems. It synthesizes a value function to represent a long-term safe set called feasible region. Early synthesis methods based on state space discretization cannot scale to high-dimensional problems, while recent methods that use neural networks to approximate value functions result in unverifiable feasible regions. To achieve both scalability and verifiability, we propose a framework for synthesizing verified neural value functions for HJ reachability analysis. Our framework consists of three stages: pre-training, adversarial training, and verification-guided training. We design three techniques to address three challenges to improve scalability respectively: boundary-guided backtracking (BGB) to improve counterexample search efficiency, entering state regularization (ESR) to enlarge feasible region, and activation pattern alignment (APA) to accelerate neural network verification. We also provide a neural safety certificate synthesis and verification benchmark called Cersyve-9, which includes nine commonly used safe control tasks and supplements existing neural network verification benchmarks. Our framework successfully synthesizes verified neural value functions on all tasks, and our proposed three techniques exhibit superior scalability and efficiency compared with existing methods.
\end{abstract}

\section{Introduction}
\label{sec: intro}
Safety is a primary concern in controller design, especially for control systems interacting with the physical world, such as autonomous driving and robot locomotion.
In these systems, safety is usually specified by inequality constraints on system states.
For example, safety constraints in a robot locomotion task require the distance between the robot and surrounding obstacles to be always positive.
Such safety constraints must be satisfied not only in a single time step but also in all time steps over an infinite horizon.
When designing a controller, it is important to know from which states it can satisfy the infinite-horizon safety constraints and from which states it cannot.
The deployment of any controller should be restricted to a set of states where long-term constraint satisfaction is ensured, which is called a \textit{feasible region}.

Hamilton-Jacobi (HJ) reachability analysis provides a formal method for computing feasible regions of control systems with safety constraints \shortcite{bansal2017hamilton}.
In HJ reachability analysis, a feasible region is represented by the zero-sublevel set of a value function, which is defined as the maximum value of the constraint function over a trajectory.
In general nonlinear systems, exactly computing the value function is difficult because it involves solving the HJ partial differential equation, which does not have a closed-form solution in most cases \shortcite{bansal2017hamilton}.
Traditional methods numerically solve the HJ equation on a grid representing a discretization of the state space \shortcite{mitchell2007comparing,mitchell2008flexible}.
These methods' computational complexity grows exponentially with state dimension, making them intractable in high-dimensional systems.
Although some techniques, such as system decomposition \shortcite{chen2016fast}, are proposed for accelerating value function computation, they only apply to some special scenarios.
Recent methods use neural networks to approximate the solution of the HJ equation by minimizing the error between the two sides of the equation \shortcite{fisac2019bridging,bansal2021deepreach}.
The error is computed on states randomly sampled in the state space and minimized by gradient-based optimization algorithms.
Although these methods scale well to high-dimensional systems, the zero-sublevel sets of their value functions are no longer guaranteed to be valid feasible regions due to neural network approximation errors.
Specifically, their zero-sublevel sets may violate two basic properties of a feasible region: \textit{constraint satisfaction} and \textit{forward invariance}.
Constraint satisfaction means all states in a feasible region satisfy the safety constraint themselves.
Forward invariance means starting from any state in a feasible region, its subsequent states can always be kept in this region by some control policy.
Violating these two properties may cause possible constraint violations, even starting from a state inside the zero-sublevel set, making the value function unreliable for safe control.

The problem of invalid feasible regions necessitates verification of neural HJ reachability value functions.
Recently, some researchers have begun to use neural network verification tools to formally verify and synthesize neural safety certificates, such as neural control barrier functions (CBFs) \shortcite{wang2023simultaneous,zhang2024exact} and neural control Lyapunov functions (CLFs) \shortcite{chang2019neural,abate2020formal}.
Similar to HJ reachability value function, these safety certificates also represent feasible regions by their zero-sublevel sets.
The difference is that these functions are not defined through equations but through certain inequality conditions.
For example, the time derivative of a CBF must be upper bounded by an extended class $\mathcal{K}$ function, and a CLF must be a positive definite function with negative time derivatives.
Existing works try to verify whether these conditions are strictly satisfied by the neural safety certificates in the entire state space.
Such problems can be transformed into standard neural network verification problems, which can be solved by existing verification tools \shortcite{liu2021algorithms}.
This verification procedure can also be embedded into neural safety certificate synthesis, resulting in verification-guided training methods \shortcite{wang2023simultaneous,zhang2024exact}.
If verification fails on a synthesized safety certificate, the found counterexamples are added to the dataset, and the safety certificate is further trained on these counterexamples.
This process is repeated until the neural safety certificate is successfully verified.
However, these methods currently only work on low-dimensional or linear systems and are difficult to scale to high-dimensional nonlinear systems and real-world control tasks.
Through our study, we discover three main challenges that restrict the scalability of these methods:
\begin{itemize}
\item \textbf{Difficulty of searching and eliminating counterexamples.}
Successful verification requires strict satisfaction of inequality conditions in all states, and not a single counterexample is allowed.
However, counterexamples become extremely sparse in high-dimensional spaces, making it difficult to find and eliminate all of them.
\item \textbf{Severe shrinkage of feasible region.}
When training neural safety certificates on counterexamples, their zero-sublevel sets tend to shrink so that the inequality conditions can be more easily satisfied.
Although slight shrinkage is sometimes acceptable, severe shrinkage can be a serious problem because it results in overly conservative control policies and poor control performance.
We find in our experiments that, in some cases, the feasible region shrinks so much that it even disappears.
\item \textbf{High computational complexity of verification.}
The computational complexity of neural network verification algorithms typically grows exponentially with the input dimension, which equals the system's state dimension.
As shown in our experiments, verifying a relatively small value network in a 4-dimensional system can take more than 2 hours on a common computing platform.
\end{itemize}

\begin{figure}
    \centering
    \includegraphics[width=\textwidth]{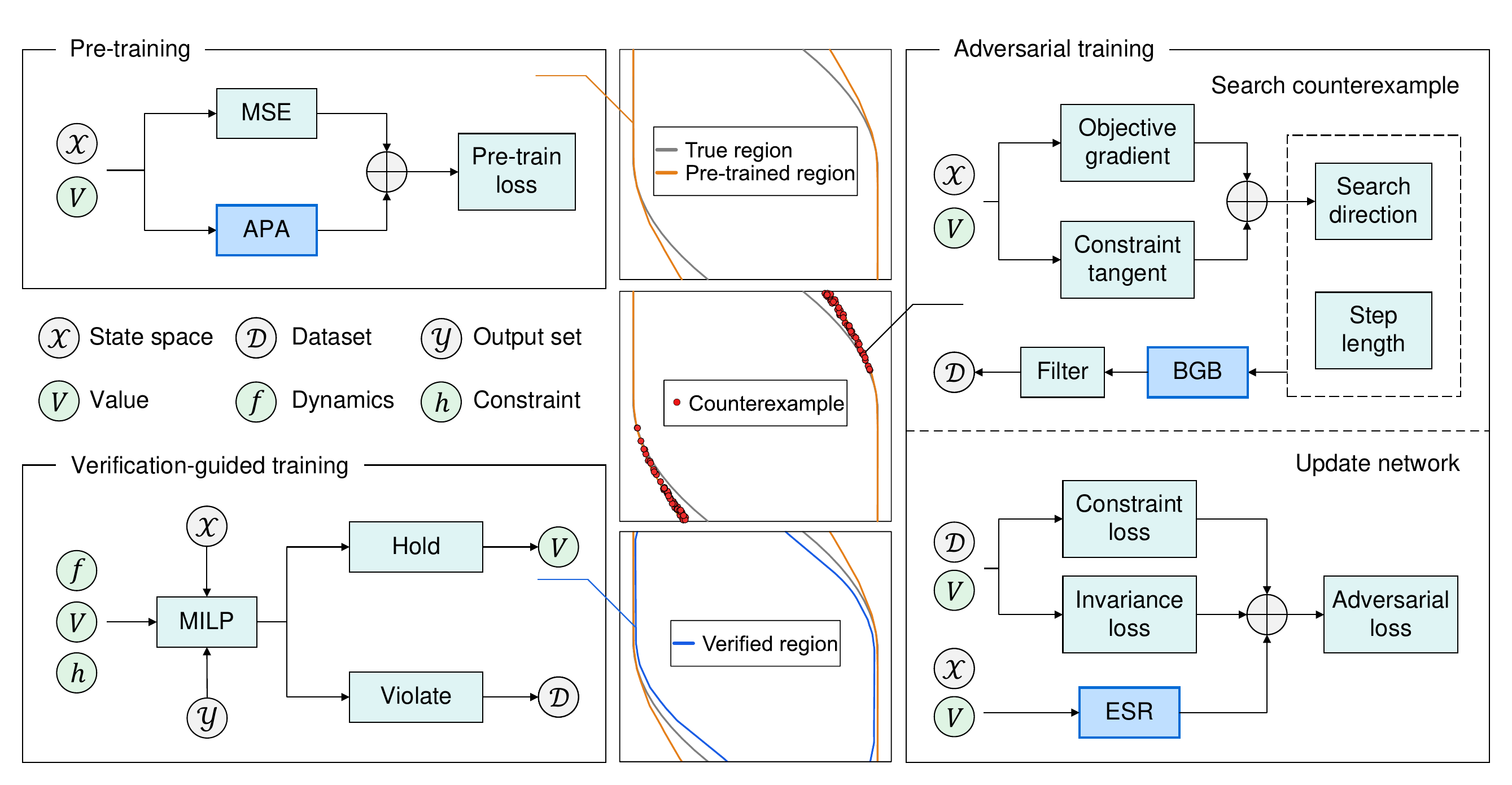}
    \caption{Neural HJ reachability value function synthesis framework. Our framework consists of three stages: pre-training, adversarial training, and verification-guided training. The circles of state space and dataset in pre-training and adversarial training mean randomly sampling states from these sets, while in verification-guided training, the state space and output set are in analytical form. Our three key techniques are highlighted in blue boxes, namely APA, BGB, and ESR. The three middle figures show the synthesis results on a 2D task Double Integrator. The pre-trained feasible region is larger than the true feasible region, indicating that it is invalid. After adversarial training and verification-guided training, the feasible region becomes valid. Details of the experiment can be found in Section \ref{sec: experiment}.}
    \label{fig: framework}
\end{figure}

The above analysis reveals a key challenge of synthesizing neural HJ reachability value functions: the trade-off between scalability and verifiability.
Traditional numerical methods based on state space discretization ensure verifiability but sacrifice scalability.
Neural HJ reachability methods scale well to high-dimensional systems but their synthesized networks are not verified.
Recent verification-guided training methods provide a promising way to synthesize verifiable neural safety certificates, but they again sacrifice scalability because of the aforementioned three challenges.
To achieve both scalability and verifiability, we propose a scalable framework for synthesizing formally verified neural HJ reachability value functions, as shown in Figure~\ref{fig: framework}.
Our framework consists of three stages: pre-training, adversarial training, and verification-guided training.
Pre-training and verification-guided training are widely used in existing neural safety certificate synthesis methods \shortcite{wang2023simultaneous,chang2019neural,abate2020formal}.
The former aims to obtain a reasonable approximation of the value function, while the latter aims to fine-tune the network on counterexamples until it becomes a valid safety certificate.
However, directly performing verification on a pre-trained value network is inefficient because a pre-trained network usually has a large number of counterexamples, while only a single one can be found in each verification step.
To improve fine-tuning efficiency, we add an adversarial training stage between them, which searches and eliminates counterexamples in a batched manner.
Verification-guided training does not start until adversarial training can hardly find any counterexamples.
Although there are many advanced methods for neural network verification, e.g., $\alpha,\beta$-CROWN \shortcite{xu2020fast,wang2021beta}, we choose a basic method: mixed integer linear programming (MILP) \shortcite{tjeng2017evaluating}.
This is because most advanced methods are designed for robust image classification problems, while we consider safety certificate synthesis problems, which have much lower input dimensions and smaller neural networks.
In such small-scale problems, those advanced methods perform even worse than the basic MILP \shortcite{wang2021beta}.
Note that this does not mean that our problem is simpler than robust image classification; in some sense, our problem is even harder as we require certain properties to hold in the entire input space instead of a small disturbance set.
With our three-stage framework in place, the next step is to find specific methods to solve the problem in each stage.
Although there are many well-studied methods for these three stages, we find that directly using state-of-the-art methods makes the framework quickly fail as the problem dimension increases, because these methods cannot effectively solve the aforementioned three challenges.
To this end, we propose three techniques to improve the scalability of our framework, each designed to address one of the three challenges. 
First, we find that counterexample search in adversarial training is difficult because existing gradient-based search methods are inefficient in searching along the boundary of feasible region.
We propose a backtracking method that rotates the search direction towards the boundary to accelerate counterexample search.
Second, we exploit the fact that constraint-satisfying states that enter the feasible region in one step are also feasible and penalize the value of these states in fine-tuning to alleviate the shrinkage of feasible region.
Third, we discover that the number of linear segments of the neural network greatly affects the computational complexity of solving MILP.
We design a regularization term for network pre-training to reduce the number of linear segments and thus accelerate verification.
Our main contributions are summarized as follows.
\begin{itemize}
\item We propose a scalable framework for synthesizing formally verified neural HJ reachability value functions.
Our framework synthesizes neural value functions from coarse to fine with high efficiency through three stages: pre-training, adversarial training, and verification-guided training.
Pre-training approximates the solution to an HJ equation by gradient descent on data samples and obtains a value network that is likely to be invalid.
Adversarial training searches counterexamples in a batched manner based on necessary and sufficient conditions for feasible region and fine-tunes the value network to eliminate counterexamples.
Verification-guided training formulates the value network verification problem as an MILP and further fine-tunes the network on counterexamples until it is verified.
\item To accelerate counterexample search in adversarial training, we propose an algorithm called boundary-guided backtracking (BGB) that efficiently searches along the boundary of feasible region.
When approaching the boundary, BGB rotates the search direction towards the tangent plane of the boundary so that larger step sizes can be taken without stepping out of the feasible region.
\item To alleviate feasible region shrinkage, we present entering state regularization (ESR) that adds a penalty term into the loss function when fine-tuning the value network.
ESR first identifies constraint-satisfying states that enter the feasible region in one step and then encourages the values of these states to be negative so that they are included in the feasible region.
\item To accelerate MILP-based verification, we design a regularization method called activation pattern alignment (APA) for pre-training of value network and dynamics network.
APA reduces linear segments of a neural network by penalizing the difference in the activation patterns of neighboring states while minimizing the loss of network's approximation ability.
\item We provide a benchmark called Cersyve-9 for neural safety certificate synthesis and verification in safe control problems, which supplements existing neural network verification benchmarks.
Cersyve-9 contains nine commonly used control tasks with various dimensions, including linear and nonlinear system dynamics and safety constraints.
Extensive experiments on Cersyve-9 demonstrate the effectiveness, scalability, and efficiency of our synthesis framework.
The code of our benchmark is available on GitHub\footnote{https://github.com/intelligent-control-lab/Cersyve.jl}.
\end{itemize}

\section{Related Works}
\subsection{Synthesis of Neural HJ Reachability Value Function}
Traditional HJ reachability analysis computes the value function by numerically solving the HJ PDE on a discretized grid of the state space.
The computational complexity of this method grows exponentially with state dimension, making it inapplicable to high-dimensional problems.
To deal with this issue, researchers have explored using neural networks to approximate the value function.

A straightforward method for learning a neural value function is to minimize the error between the two sides of the HJ equation by gradient descent on state samples \shortcite{darbon2020overcoming,bansal2021deepreach}.
However, this method is hard to converge because the HJ equation does not yield a contraction mapping and thus does not satisfy the convergence conditions of fixed point iteration.
In practice, this method typically requires specific initialization of the value function or rely on certain neural network architectures to converge to the correct solution.
\shortciteA{fisac2019bridging} solve this problem by introducing a discount factor into the value function, modifying the original maximum constraint formulation of HJ reachability to a maximum discounted constraint formulation.
Under the discounted formulation, the HJ equation also changes to a discounted version, which yields a contraction mapping and enables convergence of fixed point iteration with an arbitrary initialization.
This makes temporal difference learning methods in reinforcement learning (RL) applicable to computing the value function.
Since then, the discounted value function has been extensively used for neural HJ reachability analysis in safe control tasks, especially when combined with RL algorithms.
For example, \shortciteA{hsu2021safety} consider reach-avoid problems and add goal information to the value function proposed by \shortcite{fisac2019bridging}.
They derive a discounted reach-avoid Bellman backup and prove that their reach-avoid Q-learning algorithm converges to an arbitrarily tight conservative approximation of the reach-avoid set.
\shortciteA{yu2023safe} establish a self-consistency condition for computing the value function of a specific policy.
They use the value function as the objective function and constraint for shield and main policies, respectively.
The value function is also used for policy switching during training and safety shield during evaluation.
\shortciteA{he2024agile} applies the method proposed by \shortcite{hsu2021safety} to train a reach-avoid value function in a quadrupedal robot locomotion task.
Their value function controls the switch between an agile policy and a recovery policy, and also guides the recovery policy as an objective function.

Despite these advancements, there is an inherent problem in approximating value function with neural networks: the zero-sublevel set of the value network may not be a valid feasible region due to approximation errors.
Specifically, the zero-sublevel set may violate the two basic properties of a feasible region: constraint satisfaction and forward invariance.
This can be problematic when using these value networks for constructing constraints in policy optimization or monitoring unsafe actions in safety filters.
With these two properties unsatisfied, even if the current state is inside the zero-sublevel set, the system may still run into a constraint-violating state sometime in the future.
This problem necessitates additional verification of the value network, which is not addressed by existing works.

\subsection{Verification of Neural Safety Certificates}
Safety certificates are real-valued functions of system state that are used to represent feasible regions and construct constraints or safety filters of control policy.
HJ reachability value function is a kind of safety certificate, and two other representative examples are CBF and CLF.
CBF and CLF are defined through certain inequality conditions which, when strictly satisfied, ensure that the zero-sublevel sets of these safety certificates are feasible regions.
Similar to HJ reachability value function, CBF and CLF can also be represented by neural networks, and the resulting neural CBF and CLF also face the problem of verification.

With the development of neural network verification tools \shortcite{liu2021algorithms}, some recent studies have begun to formally verify the inequality conditions of neural CBF and CLF.
For example, \shortciteA{zhang2024exact} first decompose a neural CBF into piecewise linear segments and then solve a nonlinear program to verify safety of each segment.
To deal with the non-differentiable ReLU activation function, they leverage a generalization of Nagumo’s theorem to prove invariance of sets with non-smooth boundaries and derive necessary and sufficient conditions for safety.
While \shortciteA{zhang2024exact} focus on verifying a given neural CBF, verification of neural safety certificates can also be combined with their training process.
This yields a verification-guided training scheme of neural safety certificates, which iterates between a learner and a verifier.
The learner updates the certificate on data samples to enforce the satisfaction of safety properties.
The verifier either verifies the certificate's validity in the entire state space or generates counterexamples and adds them to the dataset for further training.
This iterative procedure terminates when no counterexample is found by the verifier, in which case the neural safety certificate is formally verified.
This training scheme is widely used for synthesizing formally verified neural CBFs \shortcite{peruffo2021automated,abate2021fossil,chen2024verification} and neural CLFs \shortcite{chang2019neural,abate2020formal,dai2021lyapunov}.
To improve the efficiency of the verifier, \shortciteA{wang2023simultaneous} leverage the Branch-and-Bound scheme to identify partitions of the state space that are not guaranteed to satisfy CBF conditions.
Additional data from these partitions are incorporated into the training dataset for further optimization.
To accelerate neural CBF training, some works exploit Lipschitz continuity property of neural CBF and use robust training techniques to ensure the satisfaction of CBF conditions \shortcite{anand2023formally,tayal2024learning}.

Despite these exploratory works, challenges remain in scaling verifiable neural safety certificate synthesis methods to high-dimensional problems.
Most existing methods only apply to control systems with less than four state dimensions \shortcite{abate2020formal,wang2023simultaneous,tayal2024learning} or special systems with four to eight dimensions whose state consists of the derivatives of the same variable, and the dynamics is described by a single scalar ordinary differential equation \shortcite{peruffo2021automated,abate2021fossil}.
\shortciteA{chang2019neural} and \shortciteA{dai2021lyapunov} synthesize verified neural Lyapunov functions on six-dimensional humanoid and quadrotor systems respectively, but their training takes hours, and the obtained feasible regions are small, which may result in overly conservative control policies.
Through our study, we discover three main challenges that restrict the scalability of these methods: 1) difficulty of searching and eliminating counterexamples, 2) severe shrinkage of feasible region, and 3) high computational complexity of verification.
We propose three techniques to mitigate these three challenges respectively, and they together significantly improve the scalability of our synthesis framework.

\section{Preliminaries}
\subsection{State Constraint and Feasible Region}
Consider a discrete-time deterministic control system:
\begin{equation}
\label{eq: system dynamics}
    x_{t+1}=f(x_t,u_t),
\end{equation}
where $x\in\mathcal{X}$ is the state, $u\in\mathcal{U}$ is the control input, $f$ is the system dynamics, and $t\in\mathbb{N}$ is the time step.
A control policy maps a state to a control input, i.e., $\pi:\mathcal{X}\to\mathcal{U}$.
The closed-loop dynamics under control policy $\pi$ is denoted as $f_\pi(x)\coloneqq f(x,\pi(x))$.
Safety constraint of the system is specified by:
\begin{equation}
\label{eq: state constraint}
    h(x_t)\le0, \forall t\in\mathbb{N},
\end{equation}
where $h$ is the constraint function and the inequalities \eqref{eq: state constraint} are called state constraints.
Before implementing a control policy, it is necessary to identify the states where the closed-loop system always satisfies the state constraints.
Such states constitute the feasible region of the policy, which is defined as follows.
\begin{definition}[Feasible region]
\label{def: feasible region}
A feasible region of policy $\pi$, denoted as $\mathrm{X}^\pi$, is a subset of the state space $\mathcal{X}$ such that $\forall x\in\mathrm{X}^\pi$, $h(x_t)\le0, t\in\mathbb{N}$, where $x_0=x$ and $x_{t+1}=f_\pi(x_t)$.
\end{definition}
From the above definition, we can derive a set of necessary and sufficient conditions for feasible region: constraint satisfaction and forward invariance.
\begin{theorem}[necessary and sufficient conditions for feasible region]
\label{thm: feasible region conditions}
$\mathrm{X}^\pi$ is a feasible region of $\pi$ if and only if
\begin{enumerate}
    \item (Constraint satisfaction) $\forall x\in\mathrm{X}^\pi, h(x)\le0$.
    \item (Forward invariance) $\forall x\in\mathrm{X}^\pi, f_\pi(x)\in\mathrm{X}^\pi$.
\end{enumerate}
\end{theorem}
The proof of Theorem \ref{thm: feasible region conditions} follows directly from Definition \ref{def: feasible region} and is omitted here.
Constraint satisfaction means all states in a feasible region satisfy the state constraint at the current time step.
Forward invariance means for any state in the feasible region, its next state under the closed-loop dynamics still lies in this region.
These two conditions are useful for determining and identifying feasible regions because they only involve a single-step state transition instead of infinite steps as Definition \ref{def: feasible region}.
They are also used as the conditions for verifying neural HJ reachability value functions in this paper.

\subsection{Hamilton-Jacobi Reachability Analysis}
Hamilton-Jacobi (HJ) reachability analysis identifies the feasible region of a control system with state constraints by computing a value function.
In a closed-loop system under a given control policy, the value function is defined as the maximum value of the constraint function in a trajectory sampled by the policy.
\begin{definition}[HJ reachability value function]
\label{def: HJ reachability value function}
The HJ reachability value function of control policy $\pi$ is defined as
\begin{equation}
    V^\pi(x)\coloneqq\max_{t\in\mathbb{N}} h(x_t),
\end{equation}
where $x_0=x$ and $x_{t+1}=f_\pi(x_t)$.
\end{definition}
A desirable property of the value function is that its zero-sublevel set is a feasible region.
\begin{theorem}
\label{thm: zero-sublevel set feasible}
The zero-sublevel set of $V^\pi$, denoted as $\mathrm{X}_{V^\pi}\coloneqq\{x\in\mathcal{X}|V^\pi(x)\le0\}$, is a feasible region of $\pi$.
\end{theorem}
\begin{proof}
We prove that $\mathrm{X}_{V^\pi}$ satisfies the two conditions in Theorem \ref{thm: feasible region conditions}.
$\forall x\in\mathrm{X}_{V^\pi}$, we have $h(x)\le\max_{t\in\mathbb{N}} h(x_t)=V^\pi(x)\le0$.
Thus, $\mathrm{X}_{V^\pi}$ satisfies constraint satisfaction.
$\forall x\in\mathrm{X}_{V^\pi}$, we have $V^\pi(f_\pi(x))=\max_{t\ge1} h(x_t)\le\max_{t\in\mathbb{N}} h(x_t)=V^\pi(x)\le0$.
Thus, $\mathrm{X}_{V^\pi}$ satisfies forward invariance.
Therefore, $\mathrm{X}_{V^\pi}$ is a feasible region of $\pi$.
\end{proof}
Theorem \ref{thm: zero-sublevel set feasible} makes the value function useful for representing feasible regions and synthesizing safe controllers.
For example, with a possibly unsafe nominal policy and a safe backup policy $\pi$, one can determine whether the nominal policy will lead to a possibly unsafe state by checking if the next state is in $\mathrm{X}_{V^\pi}$.
Specifically, if the next state is in $\mathrm{X}_{V^\pi}$, it is safe as $\pi$ can keep the state always in $\mathrm{X}_{V^\pi}$.
If the next state is not in $\mathrm{X}_{V^\pi}$, it is possibly unsafe as $\pi$ has no safety guarantee outside $\mathrm{X}_{V^\pi}$, and we should replace the nominal policy with $\pi$ to compute a safe action.
This zero-sublevel set property is one of the most fundamental properties of safety certificates.
In addition to HJ reachability value function, other safety certificates, such as CBF and CLF, also have similar properties.

In a stochastic control system, the HJ reachability value function satisfies an HJ PDE \shortcite{bansal2017hamilton}.
In a deterministic closed-loop system, the minimum and maximum operators on control inputs and disturbances in the HJ PDE can be omitted, resulting in a simplified equation called the risky self-consistency condition.
\begin{theorem}[Risky self-consistency condition]
Let $V^\pi$ be the value function of $\pi$, $\forall x\in\mathcal{X}$, we have
\begin{equation}
\label{eq: self-consistency condition}
    V^\pi(x)=\max\{h(x),V^\pi(f_\pi(x))\}.
\end{equation}
\end{theorem}
The risky self-consistency condition is a recursive relationship between the values of the previous and subsequent states.
In RL, the value function of reward also has a similar self-consistency condition.
Here, the name ``risky" distinguishes the HJ reachability value function from the reward value function, reflecting the former's relationship to safety constraints.
With the risky self-consistency condition, we can compute the value function by solving equation \eqref{eq: self-consistency condition}.
However, this equation does not have an analytical solution in most cases, and traditional numerical methods based on state space discretization cannot scale to high-dimensional systems.
This necessitates using a neural network to represent the value function and approximate the solution of equation \eqref{eq: self-consistency condition}.

\subsection{Synthesizing Neural HJ Reachability Value Function}
A straightforward method for fitting a neural network to the solution of \eqref{eq: self-consistency condition} is to minimize the error between the two sides of the equation.
However, this method is hard to converge because equation \eqref{eq: self-consistency condition} does not yield a contraction mapping.
To create a contraction mapping, \shortciteA{fisac2019bridging} modify the risky self-consistency condition to a discounted version:
\begin{equation}
\label{eq: discounted self-consistency condition}
    V^\pi(x)=(1-\gamma)h(x)+\gamma\max\{h(x),V^\pi(f_\pi(x))\},
\end{equation}
where $\gamma\in(0,1)$ is a discounted factor.
They prove that equation \eqref{eq: discounted self-consistency condition} induces a contraction mapping on the value function space under the infinity norm \shortcite{fisac2019bridging}.
This ensures that a fixed point iteration converges to the unique solution of equation \eqref{eq: discounted self-consistency condition}, which is a discounted version of the value function.
As the discount factor $\gamma$ approaches one, the solution of \eqref{eq: discounted self-consistency condition} approaches that of \eqref{eq: self-consistency condition}, which is the original value function.

Suppose we use a feedforward neural network $V_\theta^\pi$ to represent the value function, where $\theta$ is the network parameters.
To approximate the solution of equation \eqref{eq: discounted self-consistency condition}, we minimize the mean squared error (MSE) between the two sides of the equation: 
\begin{equation}
\label{eq: risky self-consistency loss}
    L_\mathrm{RSC}(\theta)=\frac{1}{N}\sum_{i=1}^N(V_\theta^\pi(x^{(i)})-((1-\gamma)h(x^{(i)})+\gamma\max\{h(x^{(i)}),V_\theta^\pi(f_\pi(x^{(i)}))\}))^2,
\end{equation}
where the subscript ``RSC" stands for risky self-consistency condition, $N$ is the number of state samples, and $x^{(i)}$ stands for the $i$-th state sample.
In practice, the states are uniformly sampled from the state space. 
However, a problem with this method is that the obtained value network may not be a valid safety certificate.
Specifically, the zero-sublevel set of $V_\theta^\pi$, denoted as $\mathrm{X}_{V_\theta^\pi}$, may not strictly satisfy the necessary and sufficient conditions for a feasible region given in Theorem \ref{thm: feasible region conditions}.
The cause for this invalidity is two-fold: 1) approximation errors of $V_\theta^\pi$ make it not exactly the solution to \eqref{eq: discounted self-consistency condition}, and 2) when the discount factor $\gamma$ is less than one, the zero-sublevel set of the solution to the discounted self-consistency condition \eqref{eq: discounted self-consistency condition} is an over-approximation of that of the solution to the original self-consistency condition \eqref{eq: self-consistency condition} \shortcite{akametalu2023minimum}.
Since the zero-sublevel set of the solution to \eqref{eq: self-consistency condition} is already the maximum feasible region of $\pi$, any of its over-approximations is not a valid feasible region.
This invalidity may lead to unsafe behaviors when using $V_\theta^\pi$ for synthesizing control policies.
To address this issue, this paper aims to synthesize neural HJ reachability value functions with verified zero-sublevel sets.
\begin{problem}
\label{pro: synthesis}
Given a control system \eqref{eq: system dynamics} and a control policy $\pi$, synthesize a neural HJ reachability value function $V_\theta^\pi$, such that its zero-sublevel set is a verified feasible region of $\pi$, i.e., it strictly satisfies the conditions of constraint satisfaction
\begin{equation}
\label{eq: constraint satisfaction condition}
    \forall x\in\mathcal{X}, V_\theta^\pi(x)\le0 \Rightarrow h(x)\le0,
\end{equation}
and forward invariance
\begin{equation}
\label{eq: forward invariance condition}
    \forall x\in\mathcal{X}, V_\theta^\pi(x)\le0 \Rightarrow V_\theta^\pi(f_\pi(x))\le0.
\end{equation}
\end{problem}
Note that the constraint satisfaction and forward invariance conditions only ensure that the zero-sublevel set of $V_\theta^\pi$ is a feasible region but do not ensure that $V_\theta^\pi$ is an exact solution to \eqref{eq: self-consistency condition}.
In fact, they are only necessary conditions for $V_\theta^\pi$ to be an exact solution.
We choose to verify these two conditions because exact verification of \eqref{eq: self-consistency condition}, which is an equation involving a neural network, is almost impossible due to neural network approximation errors.
In addition, since any safety certificate represents a feasible region by its level set, these two conditions can also be used to verify other safety certificates, such as CBF and CLF, with minor modifications.

\subsection{Verifying Neural Network via Mixed Integer Linear Programming}
Verification of a neural network is to check whether the network's output lies in a specific output set for all inputs in a given input set \shortcite{liu2021algorithms}.
Mathematically, let $\mathcal{X}$ be the input set, $\mathcal{Y}$ be the output set, and $\mathrm{NN}(\cdot)$ be the neural network.
A verification problem requires to check whether the following assertion holds:
\begin{equation}
\label{eq: verification problem}
    \forall x\in\mathcal{X}, y=\mathrm{NN}(x)\in\mathcal{Y}.
\end{equation}
In this paper, we consider the case where $\mathrm{NN}(\cdot)$ is a feedforward neural network with ReLU activation functions.
In this case, $\mathrm{NN}(\cdot)$ is a piecewise linear function, and the equality $y=\mathrm{NN}(x)$ can be encoded as a set of linear and integer constraints \shortcite{tjeng2017evaluating}.
Moreover, if the input set $\mathcal{X}$ and the complement of the output set $\mathcal{Y}$ can be expressed by a finite number of linear constraints, e.g., $\mathcal{X}$ and the complement of $\mathcal{Y}$ are polytopes, assertion \eqref{eq: verification problem} can be checked by solving a mixed integer linear programming (MILP):
\begin{equation}
\label{eq: MILP}
    \mathrm{find} \ x, \quad \mathrm{s.t.} \ x\in\mathcal{X}, y\notin\mathcal{Y}, y=\mathrm{NN}(x).
\end{equation}
Problem \eqref{eq: MILP} tries to find a counterexample in the input set such that the corresponding output of the neural network is not in the output set.
If the problem is feasible, the property to verify is violated and a counterexample is found.
If the problem is infeasible, the property holds.
Note that problem \eqref{eq: MILP} is a feasibility problem, i.e., an optimization problem without an objective function.
It is also possible to include an objective function in \eqref{eq: MILP}, with common examples like maximum violation and minimum disturbance \shortcite{liu2021algorithms}.
In this paper, since we only focus on whether the property holds or not and have no preference for counterexamples, we omit the objective function to simplify the problem.

\section{Method}
This section formally introduces our neural HJ reachability value function synthesis framework.
We first provide an overview of our synthesis framework and then detail three key techniques that significantly improve its scalability.

\subsection{Overview}
\label{sec: overview}
Our framework consists of three stages: pre-training, adversarial training, and verification-guided training.
Pre-training learns a value network without verification, which probably violates the feasible region conditions.
Adversarial training efficiently searches for counterexamples and eliminates most of them by fine-tuning the network.
Verification-guided training finds remaining counterexamples by solving the MILP and further fine-tunes the network until the feasible region conditions are verified.
Pre-training is performed first and is separate from the other two stages.
Adversarial training is performed next, and if no counterexamples are found in a certain number of iterations, verification-guided training starts.
If a counterexample is found, we return to adversarial training until the next time verification-guided training is triggered.
This process is repeated until verification succeeds, at which point we have synthesized a valid neural value function.
Compared with most existing neural safety certificate synthesis methods that include pre-training and verification-guided training \shortcite{wang2023simultaneous,chang2019neural,abate2020formal}, our framework adds an adversarial training stage between them.
This is because a pre-trained value network usually has a large number of counterexamples, while only a single one can be found in each verification step, making verification-guided training inefficient.
In contrast, adversarial training searches counterexamples in a batched manner, making the search much less expensive than that in verification-guided training in terms of computation.
First using adversarial training to reduce counterexamples to a small number and then performing verification-guided training greatly improves fine-tuning efficiency.

In pre-training, we learn a value network by minimizing the following loss function.
\begin{equation}
\label{eq: pre-training loss}
    L_\mathrm{pre}(\theta)=L_\mathrm{RSC}(\theta)+L_\mathrm{APA}(\theta),
\end{equation}
where $L_\mathrm{APA}(\theta)$ is a regularization term computed by activation pattern alignment (APA), which reduces the number of linear segments of a neural network to accelerate verification.
This technique will be detailed in Section \ref{sec: APA}.
The pre-trained value network is not verified and may violate the feasible region conditions.
We fine-tune the network in the next two stages to make its zero-sublevel set a verified feasible region.

In adversarial training, we first search for counterexamples, i.e., states that violate the feasible region conditions.
For simplicity of narration, we call the counterexamples of the constraint satisfaction condition the \textit{constraint counterexamples} and those of the forward invariance condition the \textit{invariance counterexamples}. 
To find these two kinds of counterexamples, we solve two corresponding optimization problems.
The optimization problem for finding constraint counterexamples is
\begin{equation}
\label{eq: constraint CE search problem}
    \max_{x\in\mathcal{X}} \ h(x), \quad \mathrm{s.t.} \ V_\theta^\pi(x)\le0,
\end{equation}
and that for finding invariance counterexamples is
\begin{equation}
\label{eq: invariance CE search problem}
    \max_{x\in\mathcal{X}} \ V_\theta^\pi(f_\pi(x)), \quad \mathrm{s.t.} \ V_\theta^\pi(x)\le0.
\end{equation}
If the optimal value of \eqref{eq: constraint CE search problem} is greater than zero, the solution is a constraint counterexample.
Similarly, if the optimal value of \eqref{eq: invariance CE search problem} is greater than zero, the solution is an invariance counterexample.
Generally, problem \eqref{eq: constraint CE search problem} and \eqref{eq: invariance CE search problem} are non-convex in both their objective functions and constraints and are thus difficult to solve.
However, to find counterexamples, we do not need to solve them exactly but only need to find feasible points with positive objective functions, i.e., $h(x)>0$ and $V_\theta^\pi(x)\le0$ for constraint counterexamples and $V_\theta^\pi(f_\pi(x))>0$ and $V_\theta^\pi(x)\le0$ for invariance counterexamples.
To achieve this goal efficiently, we propose a gradient-based search method called boundary-guided backtracking (BGB), which will be detailed in Section \ref{sec: BGB}.

After obtaining counterexamples, we store them in a dataset for fine-tuning the value network.
In each iteration, we randomly sample the two kinds of counterexamples from the dataset and minimize their corresponding loss functions computed according to the feasible region conditions.
The loss function for constraint counterexamples is
\begin{equation}
\label{eq: constraint loss}
    L_\mathrm{con}(\theta)=\frac{1}{N_\mathrm{con}} \sum_{i=1}^{N_\mathrm{con}} -V_\theta^\pi(x_\mathrm{con}^{(i)}),
\end{equation}
and that for invariance counterexamples is
\begin{equation}
\label{eq: invariance loss}
    L_\mathrm{inv}(\theta)=\frac{1}{N_\mathrm{inv}} \sum_{i=1}^{N_\mathrm{inv}} V_\theta^\pi(f_\pi(x_\mathrm{inv}^{(i)}))-V_\theta^\pi(x_\mathrm{inv}^{(i)}),
\end{equation}
where $N_\mathrm{con}$ and $N_\mathrm{inv}$ are the numbers of constraint and invariance counterexamples respectively, and $x_\mathrm{con}^{(i)}$ and $x_\mathrm{inv}^{(i)}$ stands for the $i$-th constraint counterexample and invariance counterexample, respectively.
We discover in our experiments that directly minimizing these two loss functions will result in severe shrinkage of the value network's zero-sublevel set.
A possible reason for the shrinkage is that minimizing \eqref{eq: invariance loss} only decreases the difference between the values of the next state and the current state, but their respective values may increase instead.
To mitigate this problem, we include an additional regularization term in the loss function of adversarial training:
\begin{equation}
\label{eq: adversarial training loss}
    L_\mathrm{adv}(\theta)=L_\mathrm{con}(\theta)+L_\mathrm{inv}(\theta)+L_\mathrm{ESR}(\theta),
\end{equation}
where $L_\mathrm{ESR}(\theta)$ is computed using entering state regularization (ESR), which will be detailed in Section \ref{sec: ESR}.

In verification-guided training, we use mixed integer linear programming (MILP) to verify the feasible region conditions or find counterexamples.
If the conditions are verified, we obtain a valid neural value function, and the algorithm ends.
If counterexamples are found, we add them to the dataset for further fine-tuning.
The core problem is how to formulate the verification of the feasible region conditions as MILPs.
In a standard verification problem \eqref{eq: verification problem}, there is only one function $\mathrm{NN}(\cdot)$, while in our problem, verification of each condition involves two functions, and their corresponding inequalities have an implication relationship.
To deal with this problem, we first concatenate the two functions into a single function with two outputs:
\begin{equation}
\label{eq: models to verify}
    \mathrm{M_{con}}(x)=
    \begin{bmatrix}
        V_\theta^\pi(x) \\ h(x)
    \end{bmatrix}, \quad
    \mathrm{M_{inv}}(x)=
    \begin{bmatrix}
        V_\theta^\pi(x) \\ V_\theta^\pi(f_\pi(x))
    \end{bmatrix}.
\end{equation}
With this concatenation, the feasible region conditions can be naturally expressed by restricting the output of the concatenated function in the complement of the second quadrant in a two-dimensional space.
Take the constraint satisfaction condition as an example, ``$V_\theta^\pi(x)\le0 \Rightarrow h(x)\le0$" means the output of $\mathrm{M_{con}}$ should not be in the second quadrant.
Thus, we can define an output set whose complement is the second quadrant so that the output constraint ``$y\notin\mathcal{Y}$" can be expressed by linear inequality constraints.
Therefore, verification of the constraint satisfaction condition can be formulated as
\begin{equation}
\label{eq: constraint MILP}
    \mathrm{find} \ x, \quad \mathrm{s.t.} \ x\in\mathcal{X}, y=\mathrm{M_{con}}(x), y_1\le0, y_2\ge0,
\end{equation}
where $y_i$ stands for the $i$-th element of $y$.
Here, we assume that the state space $\mathcal{X}$ is a hyperrectangle, which is true in most cases.
Then, constraint $x\in\mathcal{X}$ in problem \eqref{eq: constraint MILP} can be expressed by linear inequalities.
Now, as long as $\mathrm{M_{con}}$ is piecewise linear, problem \eqref{eq: constraint MILP} is an MILP.
Since $V_\theta^\pi$ is a piecewise linear neural network, $\mathrm{M_{con}}$ is piecewise linear when the constraint function $h$ is piecewise linear.
Here, we assume that $h$ is piecewise linear, and the nonlinear case will be left for future work\footnote{One possible solution is to approximate a nonlinear $h$ using piecewise linear functions, e.g., Taylor model.}.
Verification of the forward invariance condition can be similarly formulated as
\begin{equation}
\label{eq: invariance MILP}
    \mathrm{find} \ x, \quad \mathrm{s.t.} \ x\in\mathcal{X}, y=\mathrm{M_{inv}}(x), y_1\le0, y_2\ge0.
\end{equation}
If the dynamics $f$ and the policy $\pi$ are both piecewise linear functions, $\mathrm{M_{inv}}$ is piecewise linear.
Similar to $h$, we assume that $f$ and $\pi$ are piecewise linear, and the nonlinear case will be left for future work.
To show how problems \eqref{eq: constraint MILP} and \eqref{eq: invariance MILP} can be encoded as MILPs, we explicitly write them in a unified standard form of MILP as follows.
Here, we view $\mathrm{M_{con}}$ and $\mathrm{M_{inv}}$ as a single feedforward neural network with ReLU activation functions.
\begin{subequations}
\begin{alignat}{1}
    \mathrm{find} \quad & x \\
    \mathrm{s.t.} \quad & \mathcal{X}_l\le x\le\mathcal{X}_u, \label{eq: input constraint} \\
    & z_0=x, \\
    & \hat{z}_j=W_j z_{j-1} + b_j, \ j=1,\dots,l, \\
    & \mathrm{if} \ \ \hat{l}_{j,k}\ge0, \ z_{j,k}=\hat{z}_{j,k}, \ j=1,\dots,l-1, \ k=1,\dots,d_j, \label{eq: ReLU encoding starts} \\
    & \mathrm{if} \ \hat{u}_{j,k}\le0, \ z_{j,k}=0, \\
    & \mathrm{otherwise}, \ \ z_{j,k}\le\hat{z}_{j,k}, \\
    & \hspace{57.5pt} z_{j,k}\ge0, \\
    & \hspace{57.5pt} z_{j,k}\le\hat{z}_{j,k}-\hat{l}_{j,k}(1-\delta_{j,k}), \ \delta_{j,k}\in\{0,1\}, \label{eq: ReLU encoding delta} \\
    & \hspace{57.5pt} z_{j,k}\le\hat{u}_{j,k}\delta_{j,k}, \label{eq: ReLU encoding ends} \\
    & \hat{z}_l=y, \\
    & y_1\le0, \ y_2\ge0,
\end{alignat}
\end{subequations}
where $\mathcal{X}_l$ and $\mathcal{X}_u$ are lower and upper bounds of $\mathcal{X}$, the inequality signs in \eqref{eq: input constraint} represent element-wise comparisons, $\hat{z}_j$ and $z_j$ are the $j$-th layers of the neural network $\mathrm{M_{con}}$ or $\mathrm{M_{inv}}$ before and after activation, $\hat{z}_{j,k}$ and $z_{j,k}$ are the $k$-th elements in the $i$-th layers, $\hat{l}_{j,k}$ and $\hat{u}_{j,k}$ are the lower and upper pre-activation bounds, $\delta_{j,k}$ is the binary variable for activation status.
Equations \eqref{eq: ReLU encoding starts} to \eqref{eq: ReLU encoding ends} are encodings of ReLU, which follow the method proposed by \shortciteA{tjeng2017evaluating}.
The pre-activation bounds $\hat{l}_{j,k}$ and $\hat{u}_{j,k}$ can be computed by any reachability-based verification method \shortcite{liu2021algorithms}, and we use CROWN \shortcite{zhang2018efficient} in this paper.

Precisely speaking, the output constraint $y_2\ge0$ in \eqref{eq: constraint MILP} and \eqref{eq: invariance MILP} should be a strict one, i.e., $y_2>0$, according to the feasible region conditions.
However, this would make the two problems no longer standard MILPs and thus difficult to solve.
Here, we relax the strict constraint to a non-strict one to maintain MILP formulations at a slight expense of completeness: the solutions to \eqref{eq: constraint MILP} and \eqref{eq: invariance MILP} are not counterexamples when $y_1\le0$ and $y_2\ge0$ simultaneously hold with equality.
Fortunately, this problem is minor in practice because the situation where these two equalities simultaneously hold is rare: it can only happen on the boundary of a feasible region.
In most cases, such a situation can be avoided by slightly shrinking the feasible region through fine-tuning.
Moreover, when solving \eqref{eq: constraint MILP} and \eqref{eq: invariance MILP} with numerical optimizers, such as Gurobi, there are likely to be numerical errors in the inequality constraints, i.e., the constraints are violated by a small amount within a certain tolerance.
Therefore, even without relaxation, $y_2$ has to be strictly less than zero with some margin so that the optimizer can consider the constraint $y_2>0$ unsatisfiable and conclude that the problem is infeasible.

In summary, our framework synthesizes a verified neural HJ reachability value function in three stages by solving four subproblems of problem \eqref{pro: synthesis}, which are defined as follows.

\begin{subproblem}[Pre-training]
\label{subpro: pre-training}
Train a value network $V_\theta^\pi$ by minimizing loss function \eqref{eq: pre-training loss}.
\end{subproblem}

\begin{subproblem}[Counterexample search]
\label{subpro: counterexample search}
Find constraint counterexamples and invariance counterexamples of a value network $V_\theta^\pi$ by solving problem \eqref{eq: constraint CE search problem} and \eqref{eq: invariance CE search problem}, respectively.
\end{subproblem}

\begin{subproblem}[Fine-tuning]
\label{subpro: fine-tuning}
Fine-tune a value network $V_\theta^\pi$ by minimizing loss function \eqref{eq: adversarial training loss} on counterexamples found by solving subproblems \ref{subpro: counterexample search} or \ref{subpro: verification}.
\end{subproblem}

\begin{subproblem}[Verification]
\label{subpro: verification}
Verify constraint satisfaction condition and forward invariance condition of a value network $V_\theta^\pi$ by solving problem \eqref{eq: constraint MILP} and \eqref{eq: invariance MILP}, respectively.
If both problems are infeasible, return the verified $V_\theta^\pi$.
If either problem is feasible, return the found counterexamples.
\end{subproblem}

\noindent
Subproblem \ref{subpro: pre-training} is first solved in the pre-training stage.
Then, subproblem \ref{subpro: counterexample search} and \ref{subpro: fine-tuning} are iteratively solved in the adversarial training stage until no counterexample is found from subproblem \ref{subpro: counterexample search}.
Next, subproblem \ref{subpro: verification} is solved in the verification guided-training stage.
When counterexamples are found from subproblem \ref{subpro: verification}, they will still be used to solve subproblem \ref{subpro: fine-tuning}, and the algorithm returns to the adversarial training stage.
This process is repeated until subproblem \ref{subpro: verification} returns a verified value function.

As mentioned above, we propose three techniques to address the three challenges of improving the scalability of neural value function synthesis and verification.
The three techniques correspond to solving subproblems \ref{subpro: pre-training}, \ref{subpro: counterexample search}, and \ref{subpro: fine-tuning}, respectively.
The following three subsections provide a detailed introduction to each technique.

\subsection{Activation Pattern Alignment}
\label{sec: APA}
A major difficulty in scaling our proposed framework is the high computational complexity of MILP-based verification.
Our experiment shows that verifying a relatively small value network in a four-dimensional system can take more than 2 hours on a common computing platform.
To alleviate this problem, we first analyze the reason for such a high computational complexity through the solving mechanism of MILPs.
For a neural network with ReLU activation functions, each ReLU unit can be either active or inactive.
To handle a neural network constraint, a binary variable is introduced for each ReLU unit to model its activation status, i.e., zero represents an inactive unit, and one represents an active unit, as shown in \eqref{eq: ReLU encoding delta} and \eqref{eq: ReLU encoding ends}.
With these binary variables, MILP problems are solved using a branch-and-bound algorithm.
The branching step divides the problem into smaller sub-problems by fixing the values of some binary variables.
The bounding step estimates the lower and upper bounds of the objective function in the sub-problems.
The computational complexity of the branch-and-bound algorithm is mainly determined by the number of branches to explore, which relies on the number of possible combinations of the binary variable values.
Since each binary variable corresponds to a ReLU unit, we can also say that the computational complexity relies on the number of possible activation patterns of the neural network.
Each activation pattern corresponds to a linear segment of the neural network.
These linear segments divide the input set into different regions, in each of which the neural network is a linear function.
Since the number of linear segments largely determines the computational complexity of solving MILPs, we aim to reduce it to accelerate verification.

For a neural network with a given structure, its number of linear segments can vary greatly depending on the network parameter values.
For a neural network with $N$ ReLU units, there is at least one linear segment and at most $2^N$ linear segments in a given input set.
To reduce the number of linear segments, we introduce a regularization method called activation pattern alignment (APA) when solving subproblem \ref{subpro: pre-training} in pre-training.
Suppose we are updating the value network on a batch of states $\{x^{(i)}\}_{i=1}^N$.
APA first adds a Gaussian noise to each state and obtains a disturbed counterpart of the state:
\begin{equation}
    \tilde{x}^{(i)}=x^{(i)}+\xi^{(i)},
\end{equation}
where $\xi^{(i)}\sim\mathcal{N}(0,\sigma^2)$.
Then, APA computes the following regularization term and adds it to the pre-training loss function \eqref{eq: pre-training loss}.
\begin{equation}
    L_\mathrm{APA}(\theta)=\alpha_\mathrm{APA}\frac{1}{N}\sum_{i=1}^N\sum_{j=1}^{l-1}\sum_{k=1}^{d_j}\frac{\min\{f^j_k(x^{(i)})\cdot f^j_k(\tilde{x}^{(i)}),0\}}{\min\{\mathrm{dropgrad}(f^j_k(x^{(i)})\cdot f^j_k(\tilde{x}^{(i)})),-\epsilon\}},
\end{equation}
where $\alpha_\mathrm{APA}>0$ is a coefficient, $l$ is the number of network layers, $d_j$ is the number of neurons in the $j$-th layer, $f^j_k$ is the value of the $k$-th neuron in the $j$-th layer before activation, and $\epsilon$ is a small constant for numerical stability.
This regularization term encourages states that are close to each other to have similar activation patterns.
It takes effect when the activation patterns of $x^{(i)}$ and its disturbed counterpart $\tilde{x}^{(i)}$ are different.
In this case, the multiplication of their pre-activation values in the numerator is penalized and driven towards a positive value so that their activation patterns become the same.
The last layer of the neural network is excluded when computing $L_\mathrm{APA}$ because it does not have a ReLU activation function and thus does not influence the number of linear segments.
The denominator of $L_\mathrm{APA}$ does not have a gradient with respect to network parameters $\theta$ and only serves as a normalization term.

Reducing the number of linear segments essentially reduces the nonlinearity of a neural network.
From this perspective, other neural network regularization methods may also achieve this goal.
A widely used regularization method is weight decay, which incorporates an L2 regularization on the network parameters into the optimization process.
Another method is the signal-to-noise ratio (SNR) loss proposed by \shortciteA{weiimprove}, which is designed to reduce the variance and improve the stability of ReLU units to mitigate performance degradation in certified training.
Compared with these methods, APA most effectively reduces linear segments while retaining the network's approximation ability to the greatest extent, as shown in Section \ref{sec: neural network regularization}.
This is because APA only takes effect when activation patterns of neighboring states are different and does not affect the specific pre-activation values when they have the same sign.

\subsection{Boundary-Guided Backtracking}
\label{sec: BGB}
In subproblem \ref{subpro: counterexample search}, we solve two constrained optimization problems \eqref{eq: constraint CE search problem} and \eqref{eq: invariance CE search problem} to find counterexamples of the feasible region conditions.
Compared with a standard adversarial training problem \shortcite{goodfellow2014explaining}, these two problems not only have a boundary constraint on the optimization variable but also have a non-convex constraint given by the zero-sublevel set of the value network.
This zero-sublevel set constraint makes these two problems difficult to solve because most existing adversarial training methods, such as the fast gradient sign method \shortcite{goodfellow2014explaining} and projected gradient descent (PGD) method \shortcite{madry2017towards}, cannot directly handle such non-convex constraints.

A straightforward method for handling the zero-sublevel set constraint is to perform a backtracking line search in each PGD iteration, resulting in PGD with backtracking (PGD-B).
Specifically, we start the search from an initial state $x_{(0)}$ randomly sampled in the zero-sublevel set.
In each iteration, we perform a backtracking line search along the gradient of the objective function followed by a projection operation until the resulting state is in the zero-sublevel set:
\begin{equation}
\label{eq: backtracking}
    x_{(k+1)}=\Pi_\mathcal{X}(x_{(k)}+\eta^s g_{\mathrm{obj},(k)}),
\end{equation}
where $\Pi_\mathcal{X}$ is the projection operator on $\mathcal{X}$, $\eta\in(0,1)$ is a constant, $s\in\mathbb{N}$ is the smallest number such that $V_\theta^\pi(x_{(k+1)})\le0$, and $g_{\mathrm{obj},(k)}$ is the unit vector of the gradient of the objective function.
Take problem \eqref{eq: constraint CE search problem} as an example,
\begin{equation}
    g_{\mathrm{obj},(k)}=\frac{\nabla_x h(x)}{\Vert\nabla_x h(x)\Vert_2}\Big|_{x=x_{(k)}}.
\end{equation}
The problem with PGD-B is that the search becomes very inefficient when approaching the boundary of the zero-sublevel set.
This is because the gradient of the objective function becomes almost vertical to the boundary, i.e., in the same direction as the gradient of the constraint function.
This is obvious for problem \eqref{eq: invariance CE search problem} because the gradients of $V_\theta^\pi(f_\pi(x))$ and $V_\theta^\pi(x)$ are very similar as long as the time step of the system is not very large.
For problem \eqref{eq: constraint CE search problem}, this phenomenon occurs when the boundary of the zero-sublevel set overlaps or is very close to that of the constrained set.
When the gradient becomes vertical to the boundary, the backtracking line search will end up in very small step sizes or even stops to avoid stepping out of the zero-sublevel set, as shown in Figure \ref{fig: PGD-B}.
It seems to be a minor problem since the counterexamples we are searching for are also located near the boundary of the zero-sublevel set; otherwise, the state cannot leave the zero-sublevel set or violate the constraint in one step unless it is already very close to the boundary.
However, practically, PGD-B would get stuck somewhere not exactly at a counterexample because the initial state is randomly chosen and counterexamples are sparsely distributed.
In this case, even if close to a counterexample, PGD-B may never find it because PGD-B can only search towards the boundary but not along it.

\begin{figure}
    \centering
    \begin{subfigure}[b]{0.45\textwidth}
        \centering
        \includegraphics[width=\textwidth]{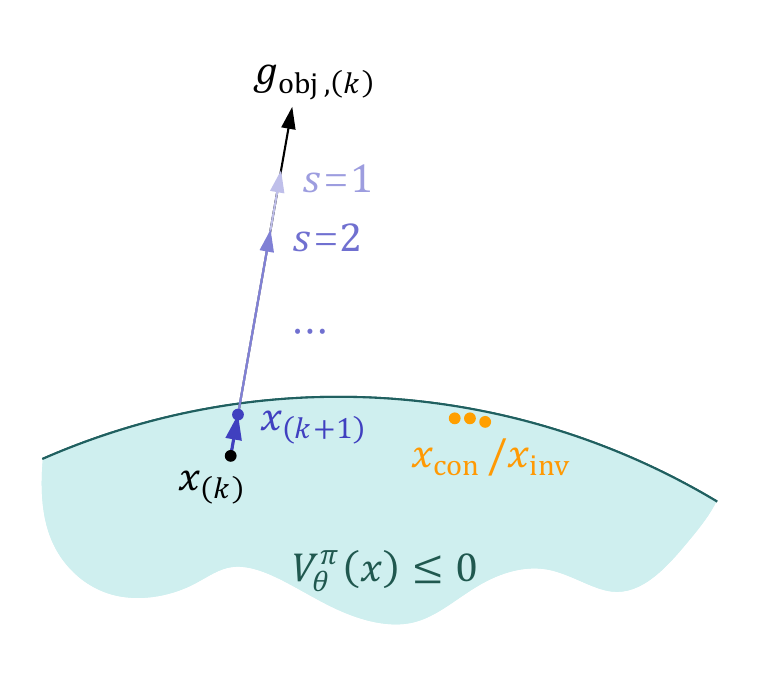}
        \caption{PGD-B}
        \label{fig: PGD-B}
    \end{subfigure}
    \begin{subfigure}[b]{0.45\textwidth}
        \centering
        \includegraphics[width=\textwidth]{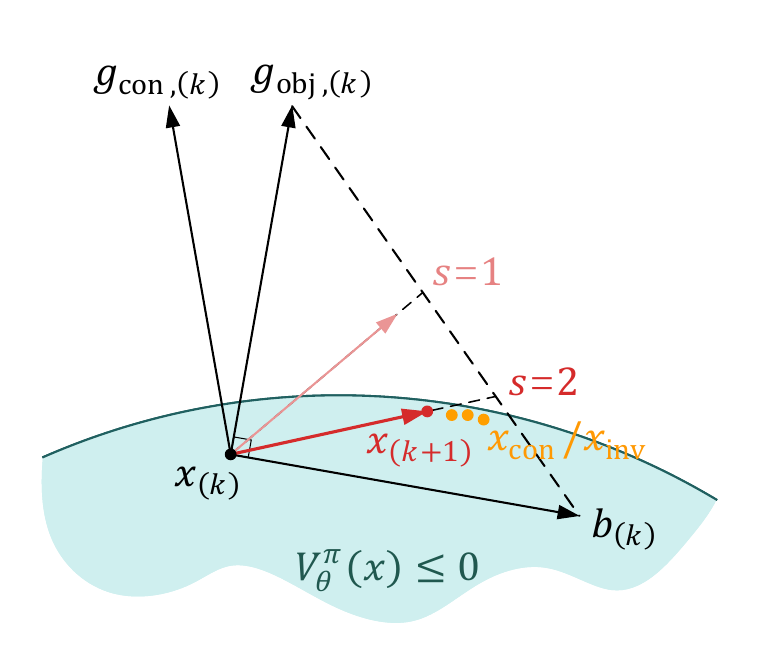}
        \caption{BGB}
        \label{fig: BGB}
    \end{subfigure}
    \caption{Search steps of PGD-B and BGB. The black dot represents the current position of the search, and the orange dots represent possible counterexamples. (a) PGD-B always searches along the gradient of the objective function and gets stuck near the boundary of the feasible region. (b) BGB rotates the search direction towards the tangent plane of the boundary and effectively search along the boundary, thus finding counterexamples more efficiently.}
\end{figure}

To solve the problem of PGD-B, we propose a boundary-guided backtracking (BGB) method that can efficiently search counterexamples along the boundary.
This is achieved by rotating the line search direction towards the tangent plane of the boundary when approaching it.
Specifically, instead of performing the line search along $g_{\mathrm{obj},(k)}$, BGB computes the search direction as a weighted sum of $g_{\mathrm{obj},(k)}$ and another unit vector $b_{(k)}$, as shown in Figure \ref{fig: BGB}.
$b_{(k)}$ is perpendicular to $g_{\mathrm{obj},(k)}$, coplanar with both $g_{\mathrm{obj},(k)}$ and $g_{\mathrm{con},(k)}$ and makes an obtuse angle with $g_{\mathrm{con},(k)}$, where $g_{\mathrm{con},(k)}$ is the unit vector of the gradient of the constraint function, i.e.,
\begin{equation}
    g_{\mathrm{con},(k)}=\frac{\nabla_x V_\theta^\pi(x)}{\Vert\nabla_x V_\theta^\pi(x)\Vert_2}\Big|_{x=x_{(k)}}.
\end{equation}
Then, $b_{(k)}$ is computed as
\begin{equation}
    b_{(k)}=\frac{\hat{b}_{(k)}}{\Vert\hat{b}_{(k)}\Vert_2}, \quad \hat{b}_{(k)}=(g_{\mathrm{obj},(k)}\cdot g_{\mathrm{con},(k)})g_{\mathrm{obj},(k)}-g_{\mathrm{con},(k)}.
\end{equation}
BGB line search is performed as
\begin{equation}
    x_{(k+1)}=\Pi_\mathcal{X}(x_{(k)}+\eta_l^s(\eta_a^s g_{\mathrm{obj},(k)}+(1-\eta_a^s)b_{(k)})),
\end{equation}
where $\eta_l$ and $\eta_a$ are backtracking coefficients for step size and search direction, respectively.
The key of BGB is the unit vector $b_{(k)}$, which determines the changing range of search direction.
We set $b_{(k)}$ vertical to $g_{\mathrm{obj},(k)}$ to ensure that the search direction always makes a sharp angle with $g_{\mathrm{obj},(k)}$ so that the objective function is always increased.
The reason for setting $b_{(k)}$ coplanar with both $g_{\mathrm{obj},(k)}$ and $g_{\mathrm{con},(k)}$ is that this is the quickest way to rotate the search direction from $g_{\mathrm{obj},(k)}$ to the tangent plane of the boundary.
Vector $b_{(k)}$ should make an obtuse angle with $g_{\mathrm{con},(k)}$ because this is the direction where the constraint function decreases.
Compared with a standard backtracking line search, BGB not only decreases the step size but also rotates the search direction towards the boundary in each iteration.
The rotation of the search direction prevents BGB from getting stuck near the boundary and enables it to search efficiently along the boundary.
This ability greatly improves the efficiency of finding counterexamples and accelerates adversarial training, as evident in Section \ref{sec: counterexample search}.

\subsection{Entering State Regularization}
\label{sec: ESR}

\begin{figure}
    \centering
    \begin{subfigure}[b]{0.49\textwidth}
        \centering
        \includegraphics[width=0.95\textwidth]{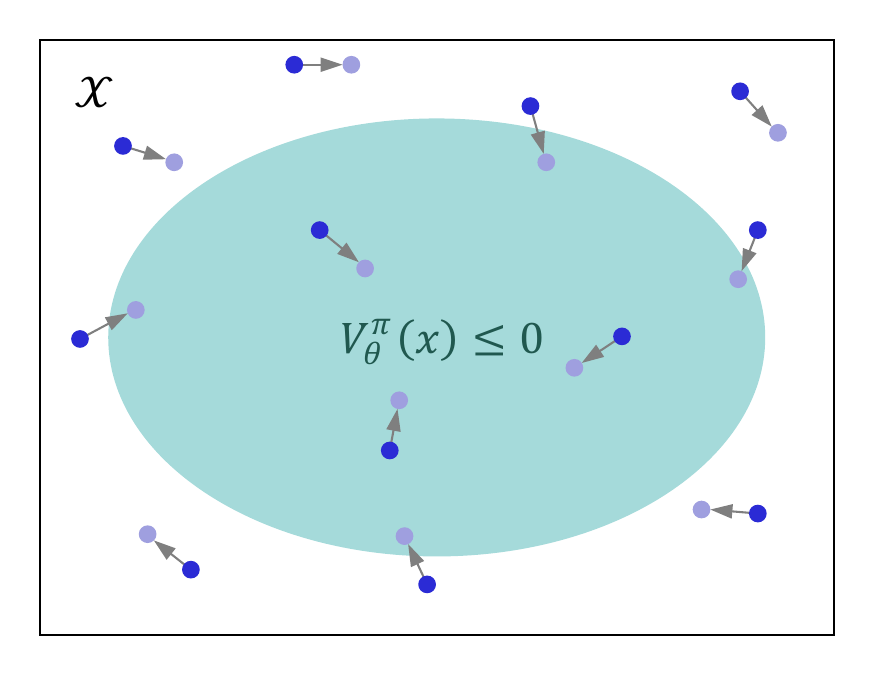}
        \caption{RSR}
        \label{fig: RSR}
    \end{subfigure}
    \begin{subfigure}[b]{0.49\textwidth}
        \centering
        \includegraphics[width=0.95\textwidth]{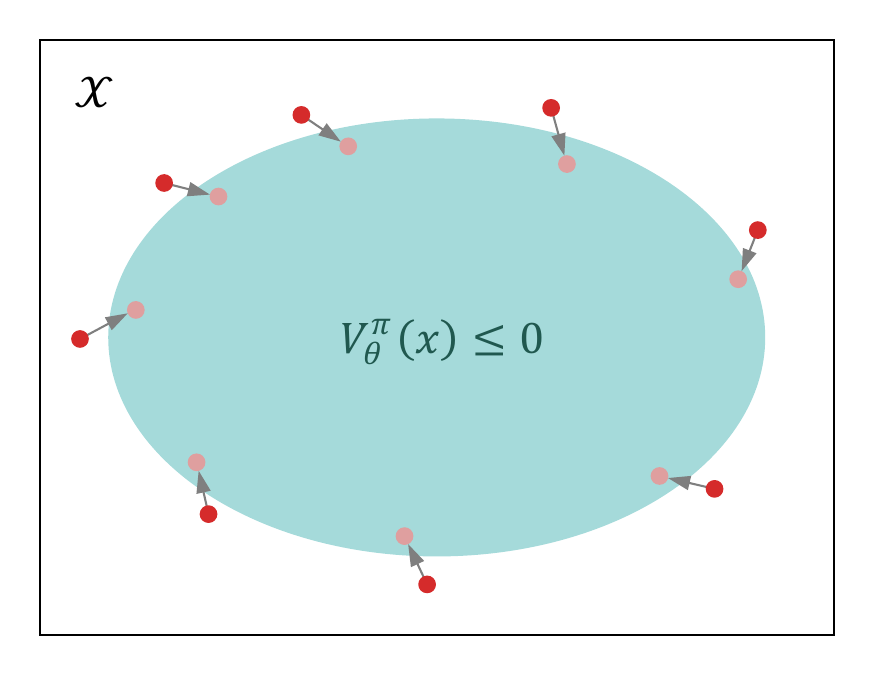}
        \caption{ESR}
        \label{fig: ESR}
    \end{subfigure}
    \caption{Regularized states of RSR and ESR. (a) The darker blue dots represent regularized states and the lighter ones are their next states. The regularized states are randomly sampled in the state space and may be infeasible. (b) The darker red dots represent regularized states and the lighter ones are their next states. The regularized states are entering states and must be feasible.}
\end{figure}

We discover in our experiments that when fine-tuning the value network in subproblem \ref{subpro: fine-tuning}, directly minimizing loss functions \eqref{eq: constraint loss} and \eqref{eq: invariance loss} on counterexamples will result in severe shrinkage of the zero-sublevel set.
Similar phenomena are also observed in other works on safety certificate synthesis \shortcite{chang2019neural,liu2023safe}.
These works deal with this problem by adding a regularization term that minimizes the output of the value network on randomly sampled states, as shown in Figure \ref{fig: RSR}.
This method is called random state regularization (RSR), and the mathematical formula of its regularization term is
\begin{equation}
    L_\mathrm{RSR}(\theta)=\alpha_\mathrm{RSR}\frac{1}{N_\mathrm{rnd}}\sum_{i=1}^{N_\mathrm{rnd}}V_\theta^\pi(x_\mathrm{rnd}^{(i)}),
\end{equation}
where $\alpha_\mathrm{RSR}$ is a coefficient, $x_\mathrm{rnd}^{(i)}$ is the $i$-th state uniformly sampled from the state space, and $N_\mathrm{rnd}$ is the number of sampled states.
The problem with RSR is that it may cause some infeasible states to be mistakenly included in the zero-sublevel set, resulting in violations of the feasible region conditions.
To avoid this problem, we only regularize states that do not compromise the satisfaction of feasible region conditions when included in the zero-sublevel set.
According to the definition of a feasible region, we can derive the following theorem, which provides a method for expanding a feasible region.

\begin{theorem}[Feasible region expansion]
\label{thm: feasible region expansion}
Let $\mathrm{X}^\pi$ be a feasible region of $\pi$. $\forall x\in\mathcal{X}\setminus\mathrm{X}^\pi$, if $h(x)\le0$ and $f_\pi(x)\in\mathrm{X}^\pi$, $\mathrm{X}^\pi\cup\{x\}$ is also a feasible region of $\pi$.
\end{theorem}

Theorem \ref{thm: feasible region expansion} tells us that if a state outside the feasible region satisfies the state constraint and enters the feasible region in one step, it can be included in the feasible region.
Such states are called \textit{entering states}.
The core idea of our regularization method, called entering state regularization (ESR), is to include entering states into the zero-sublevel set of the value network, as shown in Figure \ref{fig: ESR}.
In model predictive control, there is a concept similar to entering state called precursor set, which is defined as the set of all states whose next state is in the current set \shortcite{borrelli2017predictive}.
The precursor set is used for computing the maximal control invariant set: start from the whole state space as the initial set and iteratively intersect the current set with its precursor set, the resulting set gradually shrinks and converges to the maximal control invariant set.
Our entering states are those in the precursor set but not in the current set, and they are used for enlarging a feasible region instead of shrinking it.
To perform ESR, we first randomly sample some states in the state space in each iteration and then filter out entering states that satisfy:
\begin{equation}
\label{eq: entering state}
    h(x)\le-\delta, \ V_\theta^\pi(x)>0, \ V_\theta^\pi(f_\pi(x))\le-\delta,
\end{equation}
where $\delta$ is a small positive constant.
The purpose of introducing this constant is to avoid the undesirable influence of the regularization on nearby states.
Specifically, due to the generalization ability of neural networks, when including an entering state into the zero-sublevel set, some of its nearby states may also be included.
These nearby states may not be entering states and may cause violation of the feasible region conditions.
By introducing $\delta$, we set a margin of entering states to the boundary of the constrained set and the zero-sublevel set, thus decreasing the probability of mistaken inclusion.
Using the filtered entering states, we compute the following regularization term and add it to the value loss function.
\begin{equation}
\label{eq: ESR}
    L_\mathrm{ESR}(\theta)=\alpha_\mathrm{ESR}\frac{1}{N_\mathrm{ent}}\sum_{i=1}^{N_\mathrm{ent}}V_\theta^\pi(x_\mathrm{ent}^{(i)}),
\end{equation}
where $\alpha_\mathrm{ESR}$ is a coefficient, $N_\mathrm{ent}$ is the number of entering states, and $x_\mathrm{ent}^{(i)}$ stands for the $i$-th entering state.
Minimizing \eqref{eq: ESR} will decrease the values of entering states until they become negative, in which case they are included in the zero-sublevel set and will no longer be identified as entering states.
It is worth mentioning that the states filtered out by \eqref{eq: entering state} are not necessarily feasible because the value network has not been verified yet.
For example, it is possible that the next state $f_\pi(x)$, which is currently in the zero-sublevel set, is excluded from the set in later iterations, making the current state $x$ also infeasible.
Therefore, this regularization method may also cause mistaken inclusion of infeasible states.
However, our method is based on the fact that the value network is pre-trained, which ensures that the zero-sublevel set does not deviate much from the feasible region.
This greatly decreases the probability of including infeasible states.
Moreover, since the zero-sublevel set usually shrinks during fine-tuning and soon becomes smaller than the feasible region, the filtered entering states are feasible in most cases.
Compared with ESR, existing regularization methods are more harmful to the feasible region conditions because they use randomly sampled states for regularization, which are more likely to include infeasible states.

\subsection{Analysis of Proposed Framework}
In this subsection, we analyze some important properties of our proposed value function synthesis framework.
First, we study the soundness and completeness of the verification in our framework.
In our problem, soundness means that when MILP \eqref{eq: constraint MILP} is infeasible, the constraint satisfaction condition \eqref{eq: constraint satisfaction condition} actually holds, and when MILP \eqref{eq: invariance MILP} is infeasible, the forward invariance condition \eqref{eq: forward invariance condition} actually holds.
Completeness means that when either of the two problems is feasible, its corresponding condition is actually violated.
According to the derivation of MILPs \eqref{eq: constraint MILP} and \eqref{eq: invariance MILP}, except for the strictness of their output constraints, their infeasibilities are equivalent to the satisfaction of conditions \eqref{eq: constraint satisfaction condition} and \eqref{eq: forward invariance condition}, respectively.
Therefore, our method is sound and complete, but its implementation is subject to floating point error induced incompleteness, as discussed in Section \ref{sec: overview}.

Next, we study the relationship between the zero-sublevel set of a verified $V_\theta^\pi$, denoted as $\mathrm{X}_{V_\theta^\pi}$, and the feasible region of $\pi$.
Since a verified $V_\theta^\pi$ satisfies conditions \eqref{eq: constraint satisfaction condition} and \eqref{eq: forward invariance condition}, which are necessary and sufficient conditions for a feasible region, $\mathrm{X}_{V_\theta^\pi}$ is a feasible region of $\pi$.
Since the feasible region satisfying conditions \eqref{eq: constraint satisfaction condition} and \eqref{eq: forward invariance condition} is not unique, $\mathrm{X}_{V_\theta^\pi}$ may not be the maximum feasible region of $\pi$ but only an under-approximation of it.
Nevertheless, being a feasible region of $\pi$, $\mathrm{X}_{V_\theta^\pi}$ already guarantees that trajectories starting from inside it and sampled under $\pi$ are safe in the long term.
This enables us to use $V_\theta^\pi$ and $\pi$ to construct a safety filter in which $V_\theta^\pi$ is a safety monitor and $\pi$ is a backup policy~\shortcite{hsu2023safety}.
Specifically, starting from inside $\mathrm{X}_{V_\theta^\pi}$, the safety filter checks in each step whether the next state after applying some nominal action is still in $\mathrm{X}_{V_\theta^\pi}$.
If this is true, the nominal action is applied.
Otherwise, the nominal action is replaced by the action computed by $\pi$, which ensures by forward invariance condition that the next state is still in $\mathrm{X}_{V_\theta^\pi}$.
This safety filter can ensure strict long-term constraint satisfaction for an arbitrary policy.

Finally, we discuss the stability of our training loop, i.e., whether it is guaranteed to find a verified value function and under what conditions will it fail.
Although our framework successfully synthesizes verified value functions on various tasks in our experiments (See Section \ref{sec: experiment}), such success is not guaranteed.
In some cases, the training loop may not be able to find a verified value function with a non-trivial zero-sublevel set, even if such a true value function exists.
The reasons for the failure of synthesis are multiple, and we list four main potential failure modes as follows.
\begin{enumerate}
\item \textbf{Instability of adversarial training.}
There are always counterexamples found in adversarial training, and they cannot be eliminated.
As a result, verification cannot start before the maximum number of iterations is exceeded.
This may be because the counterexample search is too inefficient, or catastrophic forgetting happens in adversarial training.
\item \textbf{Instability of verification-guided training.}
No counterexamples can be found in adversarial training, but verification always fails.
This indicates that counterexamples exist and the search algorithm is not effective enough to find them.
\item \textbf{Inefficient verification.}
Verification takes so long that the time limit is exceeded.
This is because the computational complexity of solving the MILPs is too high.
\item \textbf{Converging to invalid local optima.}
The value function becomes all positive in the state space, and its zero-sublevel set becomes empty.
This is because too many feasible states are mistakenly excluded from the zero-sublevel set during fine-tuning, causing the set to shrink so much that it disappears.
\end{enumerate}
In essence, these failure modes all stem from the three challenges to improve scalability mentioned in Section \ref{sec: intro}.
While our proposed three techniques address these challenges to some extent, they persist as the problem's dimension increases.

\section{Experiments}
\label{sec: experiment}
Through experiments, we aim to answer the following questions:
1) Can our proposed framework synthesize verified neural HJ reachability value functions on different safe control tasks?
2) Can APA accelerate verification, and how does it perform compared with other neural network regularization methods?
3) Can BGB accelerate counterexample search, and how does it perform compared with other search methods?
4) Can ESR enlarge feasible region, and how does it perform compared with other feasible region regularization methods?
We answer the first question by testing our framework on nine commonly used safe control tasks.
We answer the remaining three questions by comparing our proposed three techniques with several existing methods.
Before that, we first introduce a neural safety certificate synthesis and verification benchmark and some implementation details of our framework.

\subsection{Cersyve-9 Benchmark}
There have been several benchmarks for neural network verification \shortcite{brix2023first}, including image classification datasets such as MNIST \shortcite{lecun2010mnist} and CIFAR \shortcite{krizhevsky2009learning}, vehicle collision prediction problem \shortcite{ehlers2017formal}, and aircraft collision avoidance system ACAS Xu \shortcite{owen2019acas}.
However, these benchmarks are incompatible with the verification of neural safety certificates because of the underlying differences between their problems and safe control tasks.
Specifically, safe control tasks require neural networks to satisfy certain properties everywhere in the state space, i.e., the input set of the verification problem is the entire state space, while existing benchmarks only consider verification in either a small disturbance set around data samples \shortcite{lecun2010mnist,krizhevsky2009learning,ehlers2017formal} or part of the state space \shortcite{owen2019acas}.
In addition, verifying neural safety certificates involves system dynamics and control policies, which requires certain conversions like \eqref{eq: constraint MILP} and \eqref{eq: invariance MILP} before it can be formulated as standard verification problems while existing benchmarks only consider verification of a single neural network.

To bridge this gap, we provide a benchmark called Cersyve-9 for neural safety \underline{Cer}tificate \underline{sy}nthesis and \underline{ve}rification in safe control tasks.
Cersyve-9 contains:
(1) Nine commonly used control tasks with state dimensions ranging from two to six, as shown in Figure \ref{fig: benchmark tasks} and Table~\ref{tab: benchmark tasks}.
These tasks include both linear and nonlinear dynamics and safety constraints.
(2) A set of neural safety certificate synthesis tools, including pre-training, adversarial training, and verification-guided training modules, as well as evaluation tools for synthesized certificates.
These tools facilitate secondary development and performance comparison of different synthesis algorithms.
(3) An MILP-based neural safety certificate verification algorithm, as well as neural value functions synthesized and verified by our framework on all nine tasks for comparing different verification algorithms.

\begin{figure}
    \centering
    \begin{subfigure}[b]{0.29\textwidth}
        \centering
        \includegraphics[width=0.9\textwidth]{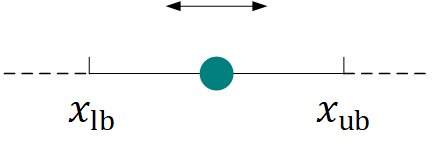}
        \caption{Double Integrator}
    \end{subfigure}
    \begin{subfigure}[b]{0.29\textwidth}
        \centering
        \includegraphics[width=0.6\textwidth]{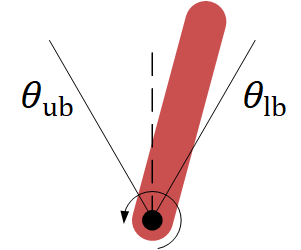}
        \caption{Pendulum}
    \end{subfigure}
    \begin{subfigure}[b]{0.29\textwidth}
        \centering
        \includegraphics[width=\textwidth]{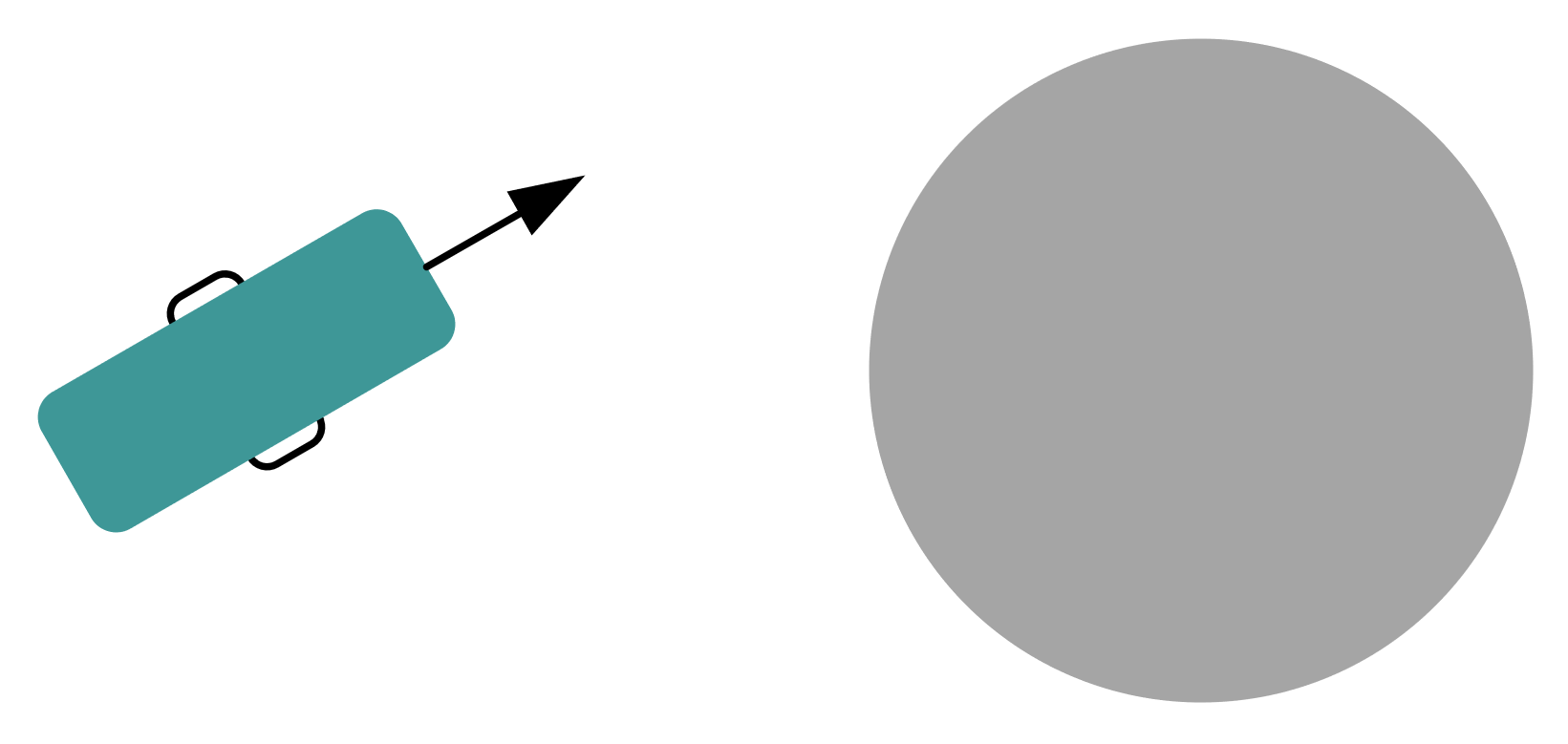}
        \caption{Unicycle}
    \end{subfigure}
    \par\bigskip
    \begin{subfigure}[b]{0.29\textwidth}
        \centering
        \includegraphics[width=0.9\textwidth]{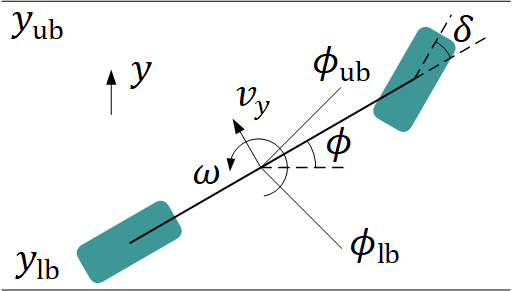}
        \caption{Lane Keep}
    \end{subfigure}
    \begin{subfigure}[b]{0.29\textwidth}
        \centering
        \includegraphics[width=0.8\textwidth]{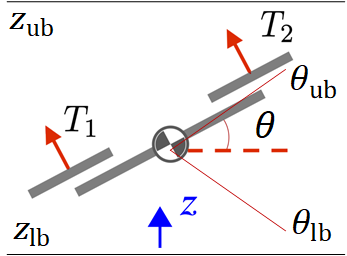}
        \caption{Quadrotor}
    \end{subfigure}
    \begin{subfigure}[b]{0.29\textwidth}
        \centering
        \includegraphics[width=0.9\textwidth]{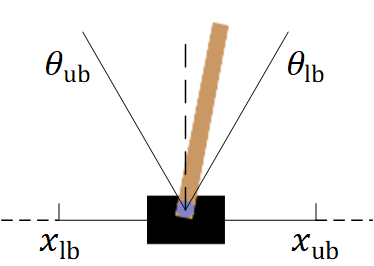}
        \caption{Cart Pole}
    \end{subfigure}
    \par\bigskip
    \begin{subfigure}[b]{0.29\textwidth}
        \centering
        \includegraphics[width=0.8\textwidth]{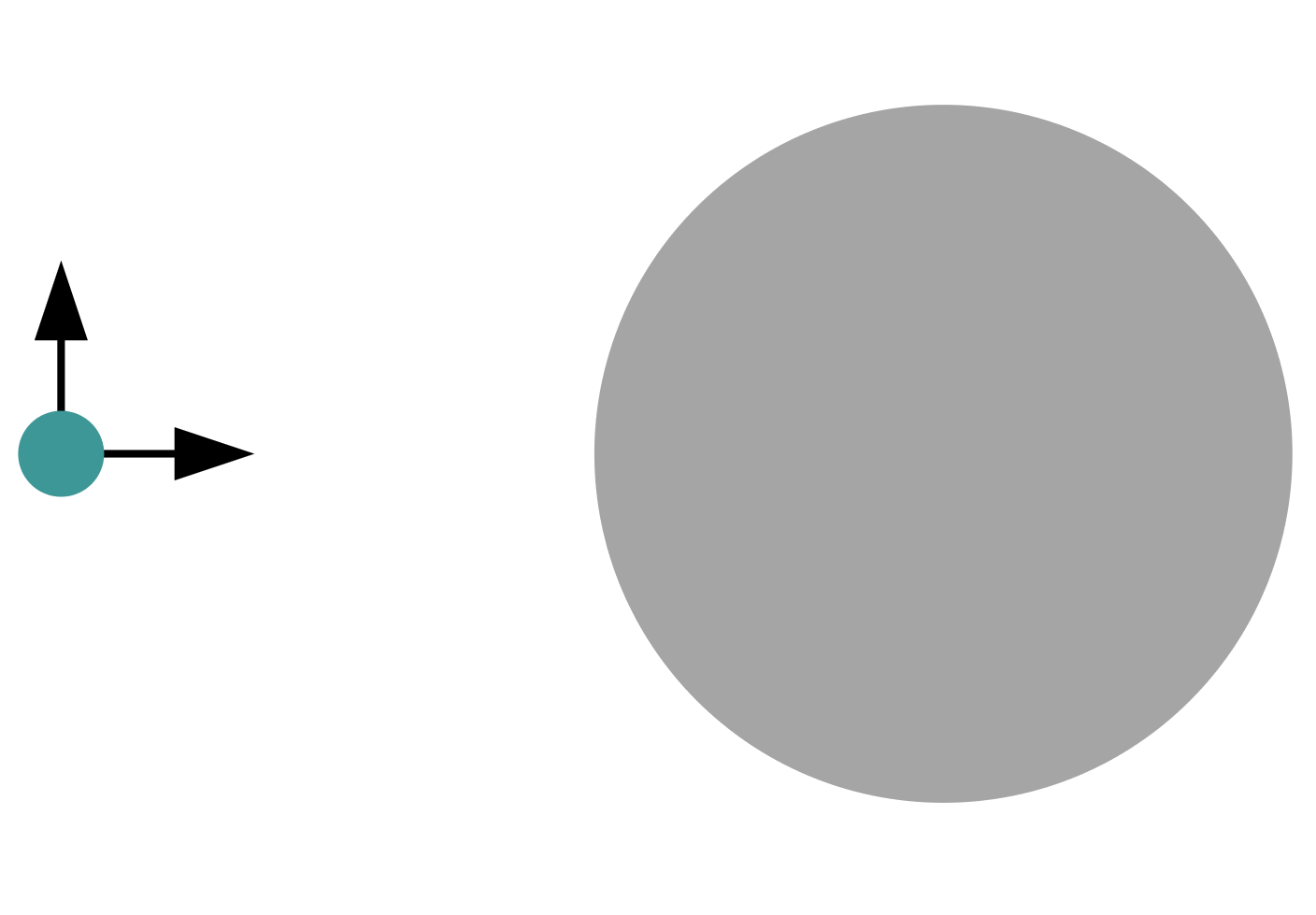}
        \caption{Point Mass}
    \end{subfigure}
    \begin{subfigure}[b]{0.29\textwidth}
        \centering
        \includegraphics[width=0.9\textwidth]{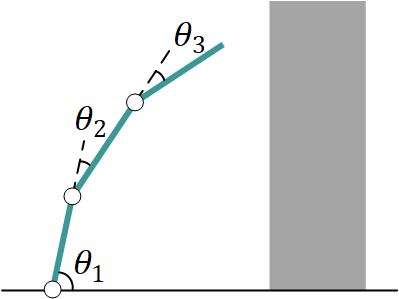}
        \caption{Robot Arm}
    \end{subfigure}
    \begin{subfigure}[b]{0.29\textwidth}
        \centering
        \includegraphics[width=\textwidth]{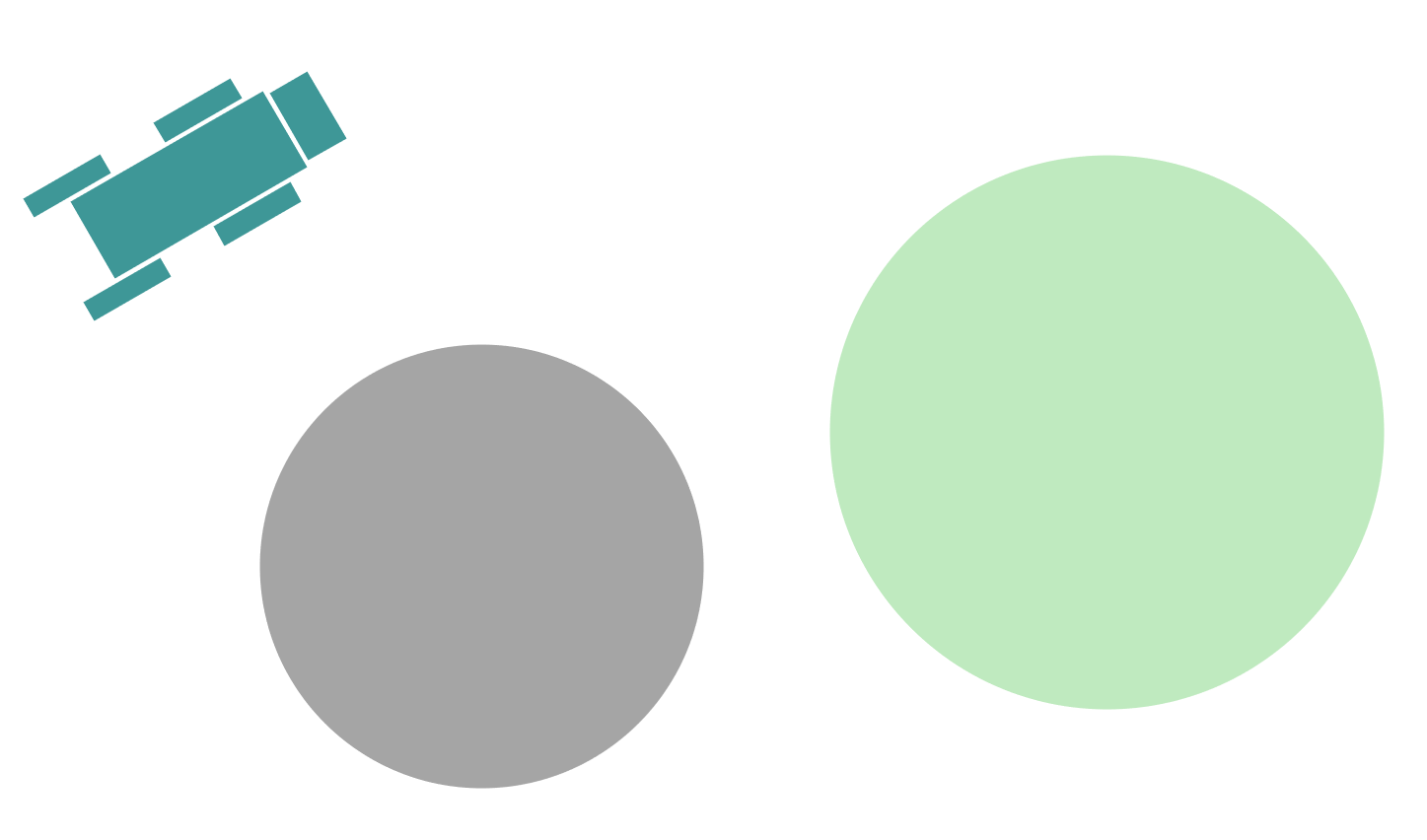}
        \caption{Robot Dog}
    \end{subfigure}
    \caption{Snapshots of safe control tasks in Cersyve-9. The grey objects in Unicycle, Point Mass, Robot Arm, and Robot Dog are obstacles. The green circle in Robot Dog is the goal.}
    \label{fig: benchmark tasks}
\end{figure}

\begin{table}[]
    \centering
    \begin{tabular}{ccccc}
    \toprule
        Task & State dim & Control dim & Dynamics & Constraint \\
    \midrule
        Double Integrator & 2 & 1 & Linear & Linear \\
        Pendulum & 2 & 1 & Nonlinear & Linear \\
        Unicycle & 3 & 2 & Nonlinear & Nonlinear \\ 
        Lane Keep & 4 & 1 & Linear & Linear \\
        Quadrotor & 4 & 2 & Nonlinear & Linear \\
        Cart Pole & 4 & 1 & Nonlinear & Linear \\
        Point Mass & 4 & 2 & Nonlinear & Nonlinear \\
        Robot Arm & 6 & 3 & Linear & Nonlinear \\
        Robot Dog & 5 & / & Nonlinear & Nonlinear \\
    \bottomrule
    \end{tabular}
    \caption{Information of safe control tasks in Cersyve-9. In Robot Dog, we directly fit the closed-loop dynamic model with a neural network, and thus, the control dimension is irrelevant to value network synthesis and verification.}
    \label{tab: benchmark tasks}
\end{table}

\subsubsection{Task descriptions}
Each task in Cersyve-9 defines state space, control input space, dynamic model, and safety constraints, which are detailed as follows.

\noindent\textbf{Double Integrator} requires stabilizing a second-order linear system to the origin under boundary constraints on the position.
The state of this task is $x=[p,\dot p]^\top\in\mathbb{R}^2$, where $p$ is the position.
The control input is $u\in\mathbb{R}$.
The state space and control input space are hyperrectangles specified by $[x_\mathrm{min},x_\mathrm{max}]$ and $[u_\mathrm{min},u_\mathrm{max}]$, respectively.
All following tasks adopt the setting of hyperrectangular state and control input spaces.
The dynamic model of this task is
\begin{equation}
\begin{aligned}
    p_{t+1}&=p_t+\dot p_t\Delta t, \\
    \dot p_{t+1}&=\dot p_t+u_t\Delta t,
\end{aligned}
\end{equation}
where $\Delta t$ is the time step.
The safety constraint is $p_\mathrm{lb}\le p\le p_\mathrm{ub}$, where $p_\mathrm{lb}$ and $p_\mathrm{ub}$ are lower and upper bounds on the position, respectively.

\noindent\textbf{Pendulum} requires stabilizing a pendulum to the upright position under boundary constraints on its angle \shortcite{brockman2016openai}.
The state of this task is $x=[\theta,\dot\theta]^\top\in\mathbb{R}^2$, where $\theta$ is the angle of the pendulum.
The control input $u\in\mathbb{R}$ is the torque applied to the pendulum.
The dynamic model is
\begin{equation}
\begin{aligned}
    \theta_{t+1}&=\theta_t+\dot\theta_{t+1}\Delta t, \\
    \dot\theta_{t+1}&=\dot\theta_t+\left(\frac{3g}{2l}\sin\theta_t+\frac{9u}{ml^2}\right)\Delta t,
\end{aligned}
\end{equation}
where $m$ is the mass of the pendulum, $l$ is its length, and $g$ is the gravitational acceleration.
The safety constraint is $\theta_\mathrm{lb}\le\theta\le\theta_\mathrm{ub}$.

\noindent\textbf{Unicycle} requires controlling a unicycle model to avoid collision with a circular obstacle.
The state of this task is $x=[v,x_\mathrm{o},y_\mathrm{o}]^\top\in\mathbb{R}^3$, where $v$ is the velocity angle of the unicycle, and $x_\mathrm{o},y_\mathrm{o}$ is the position of the obstacle in the unicycle frame.
The control input is $u=[a,\omega]^\top\in\mathbb{R}^2$, where $a$ is the acceleration of the unicycle, and $\omega$ is its angular velocity.
The dynamic model is
\begin{equation}
\begin{aligned}
    v_{t+1}&=v_t+a_t\Delta t, \\
    x_{\mathrm{o}(t+1)}&=(x_{\mathrm{o}t}-v_t\Delta t)\cos(\omega_t\Delta t)+y_{\mathrm{o}t}\sin(\omega_t\Delta t), \\
    y_{\mathrm{o}(t+1)}&=y_{\mathrm{o}t}\cos(\omega_t\Delta t)-(x_{\mathrm{o}t}-v_t\Delta t)\sin(\omega_t\Delta t).
\end{aligned}
\end{equation}
The safety constraint is $\sqrt{x_\mathrm{o}^2+y_\mathrm{o}^2}\ge r_\mathrm{o}$, where $r_\mathrm{o}$ is the radius of the obstacle.

\noindent\textbf{Lane Keep} requires keeping a 2DOF vehicle dynamics model in a straight line under boundary constraints on its lateral position and heading angle.
The state of this task is $x=[y,\phi,v_y,\omega]^\top\in\mathbb{R}^4$, where $y$ is the lateral position of the vehicle, $\phi$ is the heading angle, $v_y$ is the lateral velocity, and $\omega$ is the angular velocity.
The control input $u=\delta\in\mathbb{R}$ is the front wheel angle.
The dynamic model is
\begin{equation}
\begin{aligned}
    y_{t+1}&=y_t+(\phi_tv_x+v_y)\Delta t, \\
    \phi_{t+1}&=\phi_t+\omega_t\Delta t, \\
    v_{y(t+1)}&=\left(1+\frac{k_1+k_2}{mv_x}\Delta t\right)v_{yt}+\left(\frac{ak_1-bk_2}{mv_x}-v_x\right)\omega_t\Delta t, \\
    \omega_{t+1}&=\frac{ak_1-bk_2}{I_zv_x}v_{yt}\Delta t+\left(1+\frac{k_1a^2+k_2b^2}{I_zv_x}\Delta t\right)\omega_t,
\end{aligned}
\end{equation}
where $k_1$ and $k_2$ are the cornering stiffness of the front and rear wheels, respectively, $a$ and $b$ are the distance from the center of gravity to the front and rear axles, respectively, $m$ is the mass of the vehicle, $I_z$ is the moment of inertia on the vertical axis, and $v_x$ is the longitudinal velocity, which is a constant value.
The safety constraints are $y_\mathrm{lb}\le y\le y_\mathrm{ub}$ and $\phi_\mathrm{lb}\le\phi\le\phi_\mathrm{ub}$.

\noindent\textbf{Quadrotor} requires controlling a 2D quadrotor to a hover position under boundary constraints on its vertical position and roll angle.
The state of this task is $x=[z,\theta,\dot z,\dot\theta]^\top\in\mathbb{R}^4$, where $z$ is the vertical position of the quadrotor, and $\theta$ is the roll angle.
The control input $u=[T_1,T_2]^\top\in\mathbb{R}^2$ represents the thrust forces exerted by the rotors.
The dynamic model is
\begin{equation}
\begin{aligned}
    z_{t+1}&=z_t+\dot z_t\Delta t, \\
    \theta_{t+1}&=\theta_t+\dot\theta_t\Delta t, \\
    \dot z_{t+1}&=\dot z_t+\left(\frac{(T_{1t}+T_{2t})\cos\theta_t}{m}-g\right)\Delta t, \\
    \dot\theta_{t+1}&=\dot\theta_t+\frac{(T_{2t}-T_{1t})d}{I_y}\Delta t,
\end{aligned}
\end{equation}
where $m$ is the mass of the quadrotor, $d$ is the diameter, and $I_y$ is the moment of inertia.
The safety constraints are $z_\mathrm{lb}\le z\le z_\mathrm{ub}$ and $\theta_\mathrm{lb}\le\theta\le\theta_\mathrm{ub}$.

\noindent\textbf{Cart Pole} requires balancing a pole on a moving cart to the upright position under boundary constraints on the cart position and the pole angle \shortcite{brockman2016openai}.
The state of this task is $x=[y,\theta,\dot y,\dot\theta]^\top\in\mathbb{R}^4$, where $y$ is the cart position, and $\theta$ is the pole angle.
The control input $u=F\in\mathbb{R}$ is the force applied to the cart.
The dynamic model is
\begin{equation}
\begin{aligned}
    y_{t+1}&=y_t+\dot y_t\Delta t, \\
    \theta_{t+1}&=\theta_t+\dot\theta_t\Delta t, \\
    \dot y_{t+1}&=\dot y_t+\frac{F+ml\dot\theta_t^2\sin\theta_t-ml\ddot\theta_t\cos\theta_t}{M}\Delta t, \\
    \dot\theta_{t+1}&=\dot\theta_t+\ddot\theta_t\Delta t,
\end{aligned}
\end{equation}
where $m$ is the mass of the pole, $M$ is the total mass of the pole and the cart, $l$ is the length of the pole, and
$$
\ddot\theta_t=\frac{3Mg\sin\theta_t-3(F+ml\dot\theta_t^2\sin\theta_t)\cos\theta_t}{(4M-3m\cos^2\theta_t)l}.
$$
The safety constraints are $y_\mathrm{lb}\le y\le y_\mathrm{ub}$ and $\theta_\mathrm{lb}\le\theta\le\theta_\mathrm{ub}$.

\noindent\textbf{Point Mass} has the same setting as Unicycle except that the model is changed to a 2D point mass.
The state of this task is $x=[v_x,v_y,x_\mathrm{o},y_\mathrm{o}]^\top\in\mathbb{R}^4$, where $v_x$ and $v_y$ are the velocities on the $x$ and $y$ axes, respectively.
The control input is $u=[a,\omega]^\top\in\mathbb{R}^2$.
The dynamic model is
\begin{equation}
\begin{aligned}
    v_{x(t+1)}&=(v_{xt}+a_t\Delta t)\cos(\omega_t\Delta t)+v_{yt}\sin(\omega_t\Delta t), \\
    v_{y(t+1)}&=v_{yt}\cos(\omega_t\Delta t)-(v_{xt}+a_t\Delta t)\sin(\omega_t\Delta t), \\
    x_{\mathrm{o}(t+1)}&=(x_{\mathrm{o}t}-v_t\Delta t)\cos(\omega_t\Delta t)+y_{\mathrm{o}t}\sin(\omega_t\Delta t), \\
    y_{\mathrm{o}(t+1)}&=y_{\mathrm{o}t}\cos(\omega_t\Delta t)-(x_{\mathrm{o}t}-v_t\Delta t)\sin(\omega_t\Delta t).
\end{aligned}
\end{equation}
The safety constraint is $\sqrt{x_\mathrm{o}^2+y_\mathrm{o}^2}\ge r_\mathrm{o}$.

\noindent\textbf{Robot Arm} requires controlling a robot arm with three joints to a target position while avoiding collision with a wall in the front.
The state of this task is $x=[\theta_1,\theta_2,\theta_3, \dot\theta_1,\dot\theta_2,\dot\theta_3]^\top\in\mathbb{R}^6$, where $\theta_1$ is the angle of the first joint, and $\theta_2$ and $\theta_3$ are the incremental angles of the second and third joints relative to their previous joints.
The control input $u=[\ddot\theta_1,\ddot\theta_2,\ddot\theta_3]^\top\in\mathbb{R}^3$ represents the angular accelerations of the three joints.
The dynamic model is obtained through forward Euler discretization and omitted here.
The safety constraint is that the length of the robot arm's projection on the horizontal axis must not exceed a specific threshold:
\begin{equation}
    \sum_{i=1}^3\cos\left(\sum_{j=1}^i\theta_j\right)\le1.5.
\end{equation}

\noindent\textbf{Robot Dog} is a robot locomotion task designed by \shortciteA{he2024agile}, which requires controlling a robot dog to reach a goal while avoiding obstacles on its way.
The state of this task is $x=[v,x_\mathrm{g},y_\mathrm{g},x_\mathrm{o},y_\mathrm{o}]^\top\in\mathbb{R}^5$, where $v$ is the velocity of the robot dog, $x_\mathrm{g},y_\mathrm{g}$ is the position of the goal, and $x_\mathrm{o},y_\mathrm{o}$ is the position of the obstacle.
In this task, we directly fit a closed-loop dynamic model with a neural network on data collected by an RL policy, and thus, the control input is irrelevant to value function synthesis and verification.
The safety constraint is $\sqrt{x_\mathrm{o}^2+y_\mathrm{o}^2}\ge r_\mathrm{o}$.

\subsubsection{Implementation details}
For tasks with linear dynamic models and constraint functions, we directly use their analytical forms for value network synthesis.
For tasks with nonlinear dynamics or constraints, we fit the nonlinear dynamics or constraints with neural networks for synthesis.
We design linear control policies with control limits for all tasks except Robot Dog, where we use a neural network policy trained using the method proposed by \shortciteA{he2024agile}.
In Robot Dog, we directly fit the closed-loop dynamic model with a neural network, and thus, the control dimension is irrelevant to value network synthesis.
In other tasks with nonlinear dynamics, we fit the open-loop dynamic models with neural networks and substitute linear control policies to obtain closed-loop dynamics.
In this way, all functions involved in verification are piecewise linear, so the verification problems can be formulated as MILPs.
In theory, we could use the original nonlinear versions of the dynamics and constraints by calling some verifiers for nonlinear cyber-physical systems \shortcite{ivanov2019verisig,tran2020nnv}.
The reason why we approximate these functions with piecewise linear neural networks is to better integrate with existing verification tools.
Our neural network approximation of dynamics and constraints did introduce model mismatch to real-world robot dynamics, but that does not affect the validity of our benchmark because we can think of synthesis and verification as inherently performed on systems with approximated dynamics and constraints.
It will be our future work to investigate exact synthesis and verification with respect to nonlinear dynamics and constraints.

We follow the practice of \shortciteA{nagabandi2018neural} to train neural network dynamic models.
We use a neural network $f_\phi(x,u)$ to parameterize the change of state in a time step, i.e., the predicted next state is $\hat{x}_{t+1}=x_t+f_\phi(x_t,u_t)$.
Training data is collected by uniformly sampling initial states in the state space and executing random control inputs at every time step.
The collected data is recorded in the form of state transition pairs, i.e., $\mathcal{D}=\{(x^{(i)},u^{(i)},x'^{(i)})\}_{i=1}^N$, where $N$ is the number of data.
To ensure the loss function weights different state elements equally, we subtract the mean and divide the standard deviation of the data.
We then add zero-mean Gaussian noises with a standard deviation of $0.01$ to all data.
We train the dynamic model by minimizing the following loss function.
\begin{equation}
    L_\mathrm{dyn}(\phi)=\frac{1}{N}\sum_{i=1}^N\Vert(x'^{(i)}-x^{(i)})-f_\phi(x^{(i)},u^{(i)})\Vert_2^2.
\end{equation}
For Robot Dog, we directly train a closed-loop dynamic model $f_\phi(x)$ with training data collected by a neural network policy.
Data preprocessing and loss function of Robot Dog are similar to those of other tasks.

We use the same neural network structure in all tasks.
Both the dynamics and value networks have two hidden layers with 32 neurons each.
The constraint network has two hidden layers with 16 neurons each.
Other hyperparameters of our framework are listed in Table \ref{tab: hyperparameters}.
We ran all experiments on an AMD Ryzen 7 5800 8-Core CPU.

\begin{table}[]
    \centering
    \begin{tabular}{llc}
    \toprule
        Stage & Hyperparameter & Value \\
    \midrule
        \multirow{8}{*}{Pre-training} & Learning rate for dynamics network & 1e-3 \\
        & Learning rate for value network & 3e-4 \\
        & Batch size & 256 \\
        & Training epochs for dynamics network & 100 \\
        & Iterations for value network & 100000 \\
        & Discount factor for value network & 0.9 \\
        & Weight decay & 1e-3 \\
        & APA coefficient for dynamics network & 0.01 \\
        & APA coefficient for value network & 1e-4 \\
        & APA constant $\epsilon$ & 1e-4 \\
        & SNR coefficients for dynamics network & (0, 1e-3) \\
        & SNR coefficients for value network & (0, 0.1) \\
        & APA \& SNR noise scale & 0.1 \\
    \midrule
        \multirow{10}{*}{Adversarial training} & Learning rate & 1e-4 \\
        & Max iteration & 100000 \\
        & Batch size for counterexample search & 1000 \\
        & PGD steps per iteration & 10 \\
        & PGD step size & 0.1 \\
        & Backtracking steps & 20 \\
        & BGB search direction discount & 0.5 \\
        & BGB step length discount & 0.8 \\
        & ESR coefficient & 0.1 \\
        & ESR sample batch size & 1000 \\
        & ESR margin $\delta$ & 0.01 \\
    \bottomrule
    \end{tabular}
    \caption{Detailed hyperparameters.}
    \label{tab: hyperparameters}
\end{table}

\subsection{Evaluation Procedure and Metrics}
To evaluate synthesis results, we consider three aspects: 1) counterexample search efficiency, 2) verification efficiency, and 3) size of feasible region.

For counterexample search efficiency, we count the number of fine-tuning iterations, total fine-tuning time, and the number of verifications.
A smaller number of fine-tuning iterations means more counterexamples are found and eliminated in each iteration, thus indicating a higher search efficiency.
Fine-tuning time is an overall evaluation of the number of fine-tuning iterations and time consumption of each iteration, the latter of which largely depends on the time consumption of counterexample search.
Fine-tuning time also counts the time of all failed verifications, i.e., all verifications except the last one.
When the number of verifications exceeds one, all verification fails except the last one.
A failed verification means that counterexamples still exist, but adversarial training can no longer find them.
Therefore, more verifications also mean that counterexample search is less inefficient.
Note that calling verification does not increase fine-tuning iterations because it is a required step of each iteration to check whether verification should be called and call it when necessary.

For verification efficiency, we count the time of the final verification that proves the feasible region conditions hold.
This time plus the fine-tuning time equals the total synthesis time of the value function.

For the size of feasible region, we compute a metrics called the true feasible rate (TFR), which is defined as the proportion of states identified as feasible in all feasible states.
In practice, TFR is approximately computed on a certain number of states randomly sampled in the state space.
To determine whether a state is feasible, we check a finite-length trajectory starting from it.
The state is considered feasible if there is no constraint violation in the trajectory.
The trajectory length is set to 100 for all tasks, which is enough to give correct feasibility results in most cases.
For Robot Dog, \shortciteA{he2024agile} synthesized a neural value function without verification and used it as a safety filter in a real-world robot dog locomotion task.
We compare the neural value functions synthesized using their method and our framework to demonstrate the necessity and effectiveness of our framework.

\subsection{Synthesis Results}
Our proposed framework successfully synthesized neural HJ reachability value functions on all nine tasks in Cersyve-9, and the results are shown in Table \ref{tab: synthesis results}.
TFRs on most tasks are greater than $0.8$, and the lowest TFR is above $0.4$, indicating that the synthesized value networks represent non-trivial feasible regions.
As state dimension increases, it generally requires more fine-tuning iterations and time to synthesize a verified value network.
This is because searching counterexamples becomes more difficult in higher-dimensional spaces.
The number of verifications also shows a similar increase with state dimension.
Another observation is that systems with nonlinear dynamics generally require more verification time.
This is because nonlinear dynamics results in more linear segments of $\mathrm{M_{inv}}$ defined in \eqref{eq: models to verify} for forward invariance verification.

\begin{table}[]
    \centering
    \begin{tabular}{cccccc}
    \toprule
        Task & FT iter (k) & FT time (s) & \# Verify & Verify time (s) & TFR \\
    \midrule
        Double Integrator & 1.2 & 28.5 & 1 & 0.5 & 0.940 \\
        Pendulum & 1.1 & 31.4 & 1 & 3.9 & 0.962 \\
        Unicycle & 6.9 & 175.0 & 1 & 11.6 & 0.911 \\
        Lane Keep & 5.9 & 184.7 & 4 & 1.4 & 0.750 \\
        Quadrotor & 1.6 & 55.3 & 1 & 2.9 & 0.906 \\
        Cart Pole & 2.4 & 83.3 & 1 & 171.5 & 0.404 \\
        Point Mass & 8.6 & 321.7 & 4 & 48.2 & 0.594 \\
        Robot Arm & 25.2 & 853.8 & 15 & 5.5 & 0.403 \\
        Robot Dog & 6.3 & 333.7 & 2 & 311.9 & 0.872 \\
    \bottomrule
    \end{tabular}
    \caption{Neural HJ reachability value function synthesis results of our proposed framework.}
    \label{tab: synthesis results}
\end{table}

To understand the effect of fine-tuning on the value network, we plot the changing curves of TFR and the cumulative number of found counterexamples during fine-tuning in Robot Arm in Figure \ref{fig: TFR and cumulative counterexamples}.
At an early stage of fine-tuning, TFR decreases quickly, and counterexamples increase quickly.
In each iteration, the value network is updated on a large number of counterexamples, excluding many of them from the zero-sublevel set.
Verification cannot be called at this stage because there are many counterexamples found in every iteration.
After 10K iterations, counterexamples can hardly be found in each iteration, and therefore verification starts to be called frequently.
At this stage, the value network is updated on only a few counterexamples, mostly found by verification, in each iteration.
As a result, the change in TFR is very small.
This stage continues until the last verification proves the feasible region conditions hold and returns a valid value network.

\begin{figure}
    \centering
    \begin{subfigure}[b]{0.49\textwidth}
        \centering
        \includegraphics[width=\textwidth]{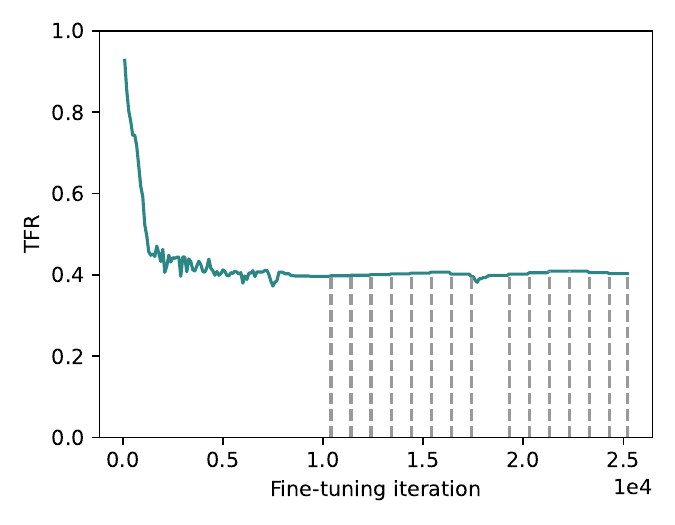}
        \caption{TFR.}
        \label{fig: TFR}
    \end{subfigure}
    \begin{subfigure}[b]{0.49\textwidth}
        \centering
        \includegraphics[width=\textwidth]{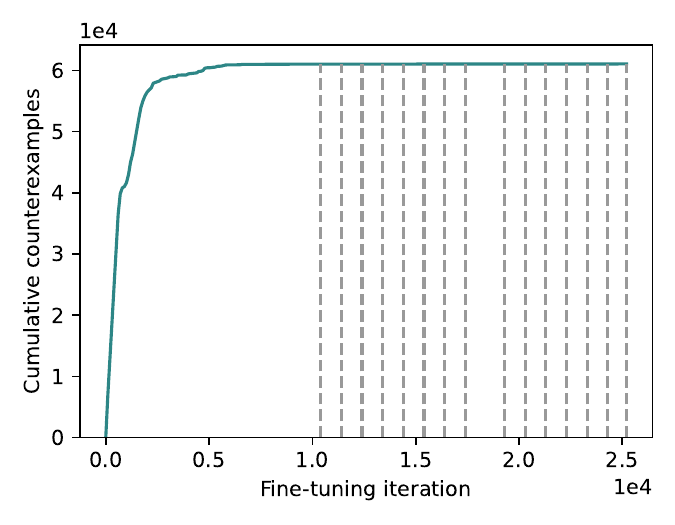}
        \caption{Cumulative counterexamples.}
        \label{fig: cumulative counterexamples}
    \end{subfigure}
    \caption{TFR and cumulative number of found counterexamples during fine-tuning in Robot Arm. The dashed gray lines stand for the iterations at which verification is called.}
    \label{fig: TFR and cumulative counterexamples}
\end{figure}

To demonstrate how counterexamples are eliminated through fine-tuning, we visualize the boundary of the zero-sublevel set and counterexamples before fine-tuning and the boundary after fine-tuning in Double Integrator in Figure \ref{fig: Double Integrator visualize}.
Before fine-tuning, there are many counterexamples near the boundary of the zero-sublevel set.
These counterexamples leave the zero-sublevel set in one step, violating the forward invariance condition and making the set not a valid feasible region.
This invalidity is also confirmed by the fact that the pre-trained region is larger than the true region, which is impossible for a valid feasible region.
After fine-tuning, all counterexamples are eliminated, and the zero-sublevel set shrinks into a valid feasible region slightly smaller than the true region.

\begin{figure}
    \centering
    \includegraphics[width=\textwidth]{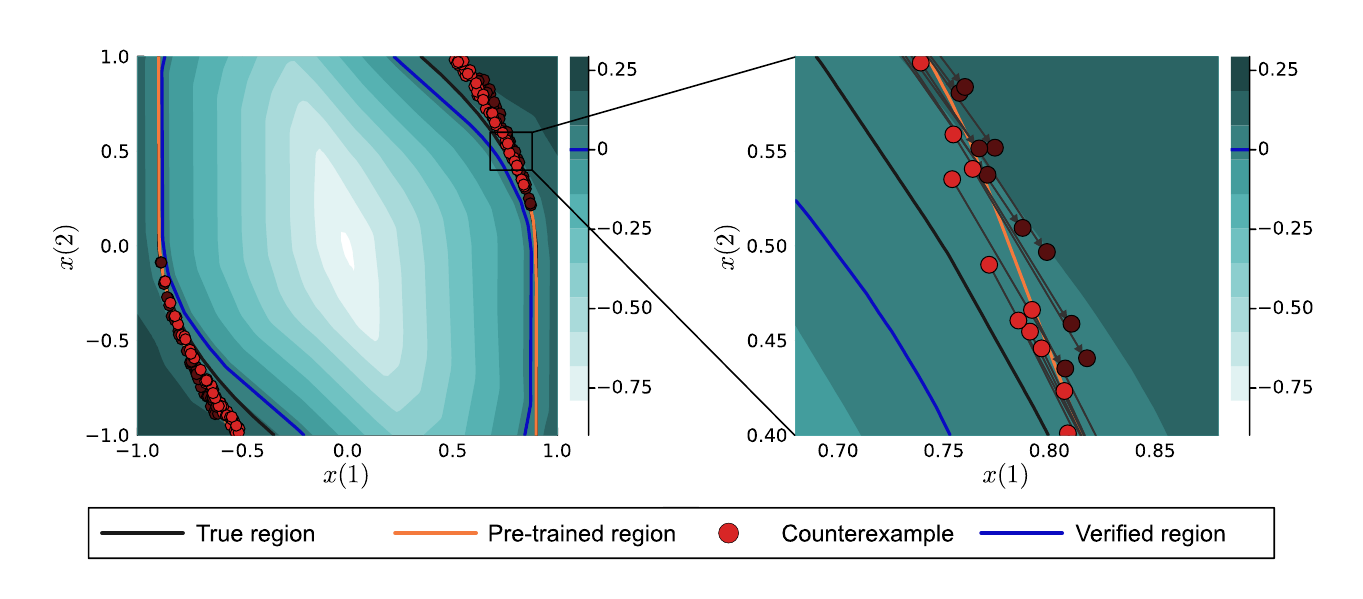}
    \caption{Regions and counterexamples in Double Integrator. The pre-trained and verified regions are zero-sublevel sets of the value networks before and after fine-tuning, respectively. The heatmap shows the contours of the value network after fine-tuning. The lighter red dots are counterexamples before fine-tuning, and the darker red dots are their next states.}
    \label{fig: Double Integrator visualize}
\end{figure}

To demonstrate the necessity and effectiveness of our method, we visualize the value trajectories and heatmaps of the value networks synthesized by \shortciteS{he2024agile} method and our method in Robot Dog, as shown in Figure \ref{fig: Robot Dog visualize}.
We choose an initial state that leads to a constraint violation and compare the values of the two networks on the trajectory.
The value of \shortciteS{he2024agile} network is negative on the initial state, and gradually increases to positive values along the trajectory.
This means that the state starts from inside the zero-sublevel set but goes out eventually, indicating that the zero-sublevel set of their value network is not a valid feasible region because it violates the forward invariance condition.
In contrast, our value network consistently outputs positive values on the whole trajectory, indicating that it correctly excludes the states from its zero-sublevel set.
Comparing the heatmaps of the two value networks, we can see that our network moves the infeasible region to the upper left.
This excludes infeasible states in the upper left from the zero-sublevel set and includes more feasible states in the lower right into the set.
As a result, the zero-sublevel set becomes a valid feasible region without a significant reduction in its size.

\begin{figure}
    \centering
    \begin{subfigure}[b]{0.49\textwidth}
        \centering
        \includegraphics[width=\textwidth]{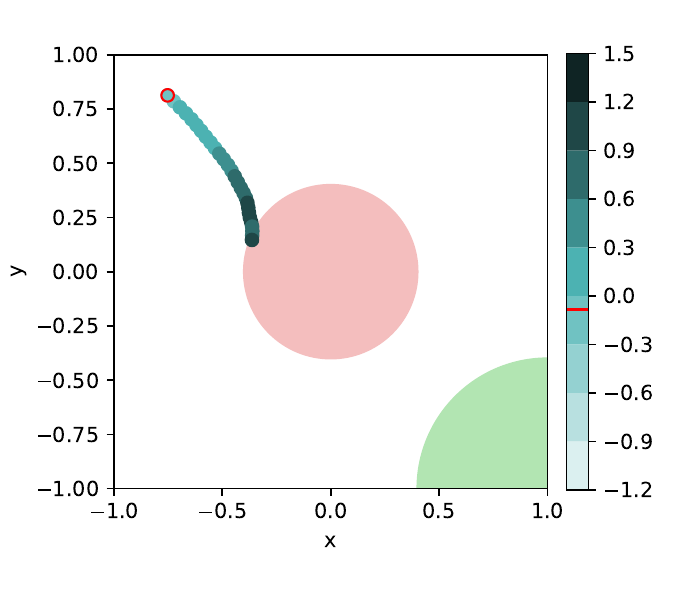}
    \end{subfigure}
    \begin{subfigure}[b]{0.49\textwidth}
        \centering
        \includegraphics[width=\textwidth]{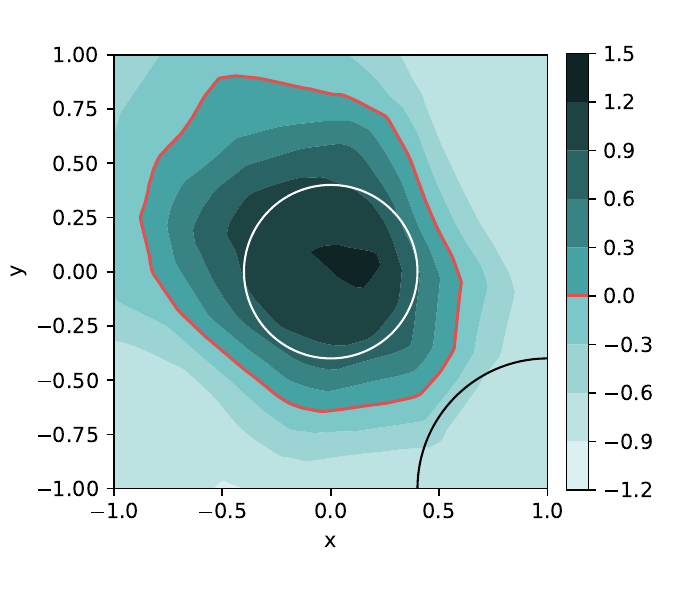}
    \end{subfigure}
    \begin{subfigure}[b]{0.49\textwidth}
        \centering
        \includegraphics[width=\textwidth]{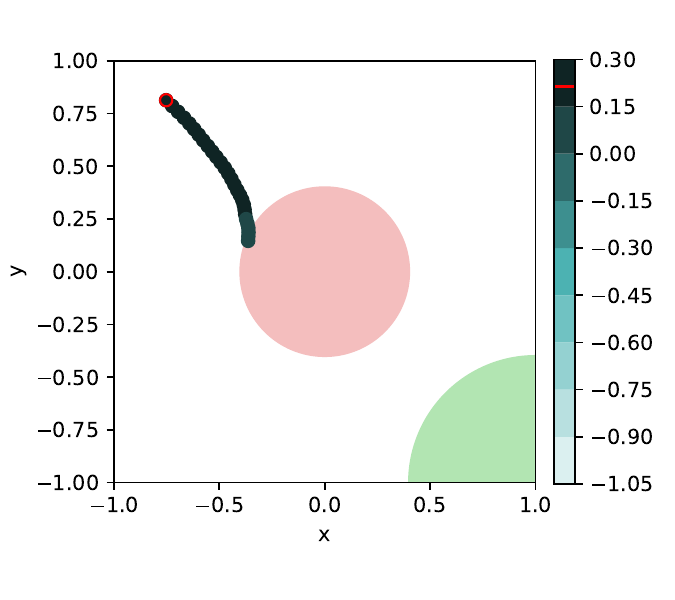}
    \end{subfigure}
    \begin{subfigure}[b]{0.49\textwidth}
        \centering
        \includegraphics[width=\textwidth]{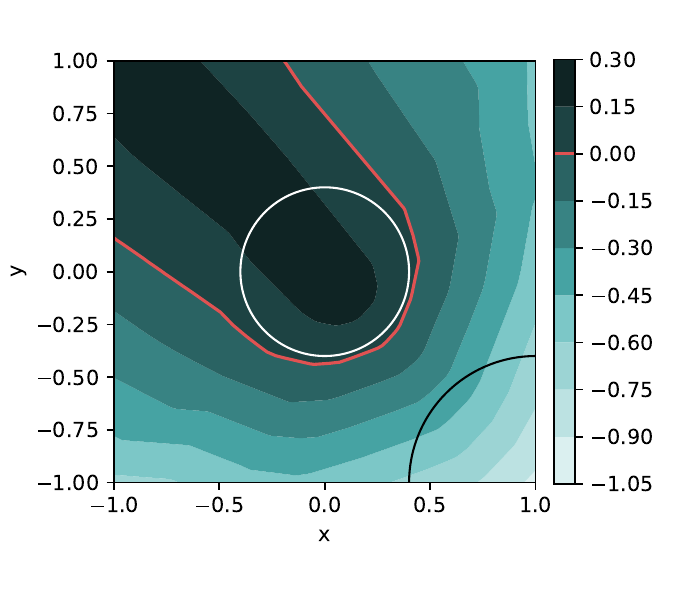}
    \end{subfigure}
    \caption{Value trajectories and heatmaps of value networks synthesized by \shortciteS{he2024agile} method (upper) and ours method (lower) in Robot Dog. In the left two figures, the red circle in the middle is an obstacle, and the green circle in the bottom right corner (a quarter shown) is the goal. The two trajectories start from the same initial state (marked with a red circle) and are sampled by the same policy. They are both truncated at a constraint-violating state.}
    \label{fig: Robot Dog visualize}
\end{figure}

\subsection{Comparison Studies}
In this subsection, we demonstrate the effectiveness of the three proposed techniques, i.e., APA, BGB, and ESR, by comparing them with several existing methods that aim to solve similar problems.

\subsubsection{Neural network regularization}
\label{sec: neural network regularization}
We compare APA with two existing neural network regularization methods, weight decay (WD) and SNR \shortcite{weiimprove}, to study its effectiveness in reducing verification time.
We compare these methods from three aspects: number of linear segments, network performance, and verification time.
First, we show the relationship between the number of linear segments and verification time.
Then, we compare the number of linear segments and network performance of the three regularization methods.
Finally, we compare the verification times of the three methods on all tasks in Cersyve-9.

Due to the branch-and-bound solving mechanism of MILP, verification time is closely related to the number of linear segments of neural networks \shortcite{wolsey2020integer}.
We use a sampling-based method to estimate the number of linear segments of a neural network.
Specifically, we uniformly sample a certain number of states in the state space, compute the neural network's activation pattern on each state, and count the number of unique activation patterns.
This gives us an underestimate of the number of linear segments, and this estimate becomes more accurate as the number of samples increases.
Figure \ref{fig: linear segment estimate} shows the relationship between the estimated number of linear segments and the number of samples.
The two have a linear relationship when the number of samples is small.
As the number of samples increases, the growth rate of linear segments decreases.
Theoretically, an infinite number of samples will give an accurate number of linear segments.
We use $10^6$ samples in our experiments to balance estimation accuracy and computational complexity.
Although this results in an underestimate, it reflects the relative number of linear segments of different methods, which is informative for comparing their verification times.
We visualize the relationship between the number of linear segments and verification time in Figure \ref{fig: linear segments vs. verification time}.
It shows that the two are approximately linearly related, which is consistent with the branch-and-bound MILP solving mechanism.
This allows us to approximately compare the verification times of different networks by comparing their number of linear segments without actually solving verification problems.

\begin{figure}
    \centering
    \begin{subfigure}[b]{0.49\textwidth}
        \centering
        \includegraphics[width=\textwidth]{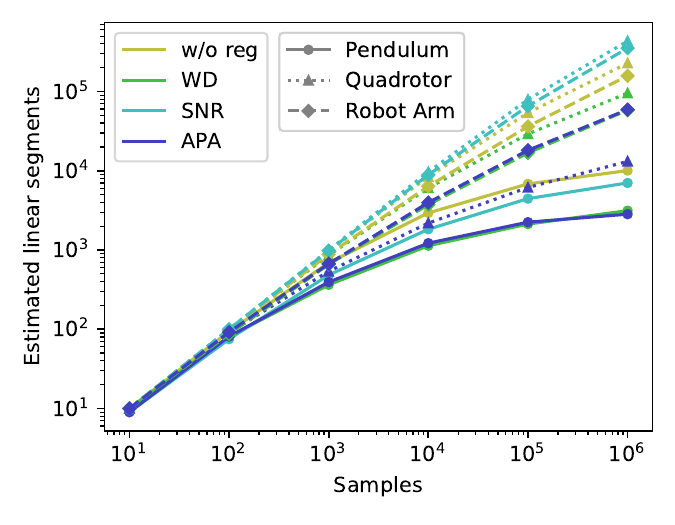}
        \caption{Linear segment estimate.}
        \label{fig: linear segment estimate}
    \end{subfigure}
    \begin{subfigure}[b]{0.49\textwidth}
        \centering
        \includegraphics[width=\textwidth]{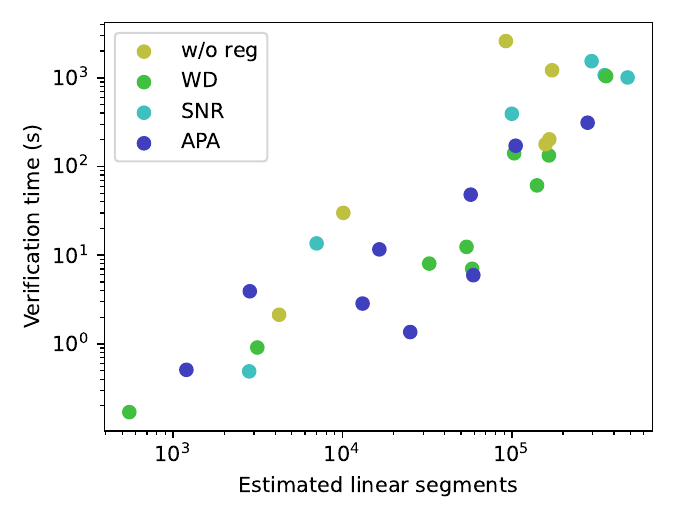}
        \caption{Linear segments vs. verification time.}
        \label{fig: linear segments vs. verification time}
    \end{subfigure}
    \caption{Estimated number of linear segments and its relationship with verification time. In Figure (b), dots with the same color represent the same method in different tasks. For a thorough comparison of verification times of different regularization methods on each task, see Table \ref{tab: verify time}.}
\end{figure}

We compare the number of linear segments of dynamics networks and value networks trained with different regularization methods, as shown in Figure \ref{fig: linear segments}.
Since regularization usually sacrifices the performance of neural networks, we also compare the performance of different regularization methods.
For dynamics networks, we compute the MSE on a test dataset for performance metrics.
For value networks, we compute TFR after fine-tuning (after the network is successfully verified) for performance metrics.
For a fair comparison, we use the same dynamics network trained with APA to train value networks in each task.
Figure \ref{fig: dynamics linear segments} shows that APA reduces linear segments of dynamics networks by about five times compared with no regularization, while WD and SNR both increase linear segments instead.
Moreover, APA has a much lower MSE compared with WD and SNR.
Figure \ref{fig: value linear segments} shows that both APA and WD significantly reduce linear segments of value networks, and APA brings a greater reduction of about four times.
SNR still results in increased linear segments of value networks.
In addition, APA has a higher TFR compared with WD and SNR, indicating larger feasible regions.
These results indicate that APA is the most effective in reducing linear segments with minimum performance sacrifice.
This is attributed to the appropriate design of APA penalty, which only takes effect when the sign of pre-activation values of neighboring states are different.
In contrast, WD and SNR are not directly targeted at making activation patterns consistent.
They penalize the network parameters at all times, resulting in large performance sacrifices and inefficiency in reducing linear segments.

\begin{figure}
    \centering
    \begin{subfigure}[b]{0.49\textwidth}
        \centering
        \includegraphics[width=\textwidth]{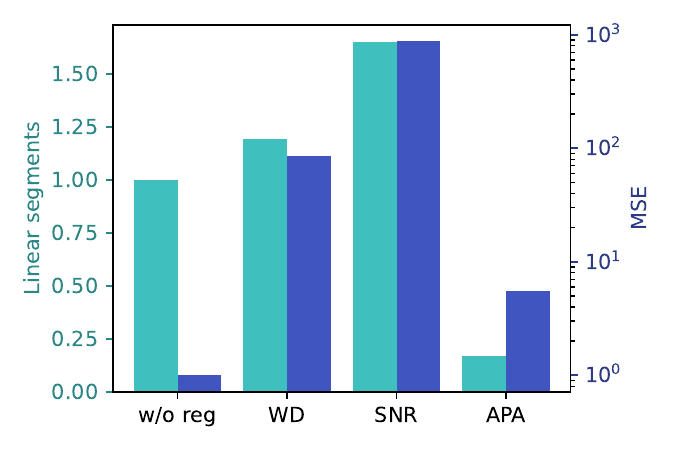}
        \caption{Dynamics networks.}
        \label{fig: dynamics linear segments}
    \end{subfigure}
    \begin{subfigure}[b]{0.49\textwidth}
        \centering
        \includegraphics[width=\textwidth]{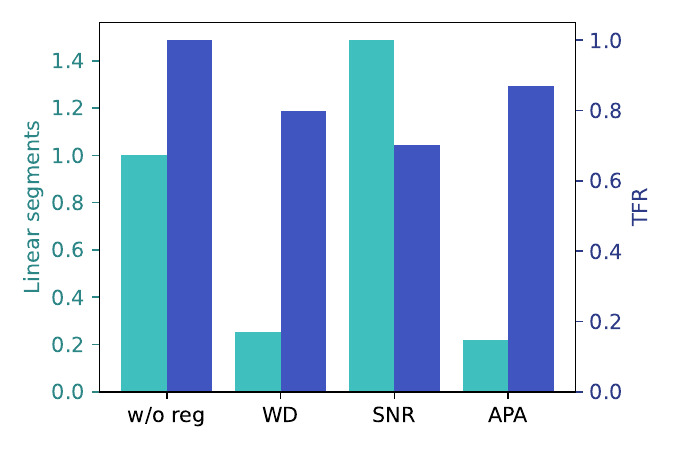}
        \caption{Value networks.}
        \label{fig: value linear segments}
    \end{subfigure}
    \caption{Number of linear segments and performance of dynamics networks and value networks trained with different regularization methods. The results of dynamics networks are averaged on five tasks with nonlinear dynamics. For each task, all scores are normalized by dividing those without regularization. The results of value networks are averaged on all tasks except Cart Pole and Robot Dog, where synthesis failed without regularization.}
    \label{fig: linear segments}
\end{figure}

The regularization strength of APA depends on the coefficient $\alpha_\mathrm{APA}$, which trades off between the number of linear segments and neural network performance.
We train dynamics and value networks under different values of $\alpha_\mathrm{APA}$ and visualize the results in Figure \ref{fig: APA coefficient}.
Figure \ref{fig: APA coefficient dynamics} shows that the number of linear segments of dynamics network quickly decreases as $\alpha_\mathrm{APA}$ increases from 0 to $10^{-3}$ and continues to decrease steadily as $\alpha_\mathrm{APA}$ increases from $10^{-3}$ to $10^{-1}$.
On the other hand, MSE also increases as $\alpha_\mathrm{APA}$ increases, and its increasing rate becomes faster.
An appropriate choice of $\alpha_\mathrm{APA}$ for dynamics network should be around $10^{-3}$ to $10^{-2}$, which balances the number of linear segments and MSE.
Figure \ref{fig: APA coefficient value} shows that both the number of linear segments of value network and TFR decreases as $\alpha_\mathrm{APA}$ increases.
The decrease rate of linear segments is relatively stable under different values of $\alpha_\mathrm{APA}$.
The decrease rate of TFR is small at first and gradually increases as $\alpha_\mathrm{APA}$ increases.
An appropriate choice of $\alpha_\mathrm{APA}$ for value network should be around $10^{-4}$.

\begin{figure}
    \centering
    \begin{subfigure}[b]{0.49\textwidth}
        \centering
        \includegraphics[width=\textwidth]{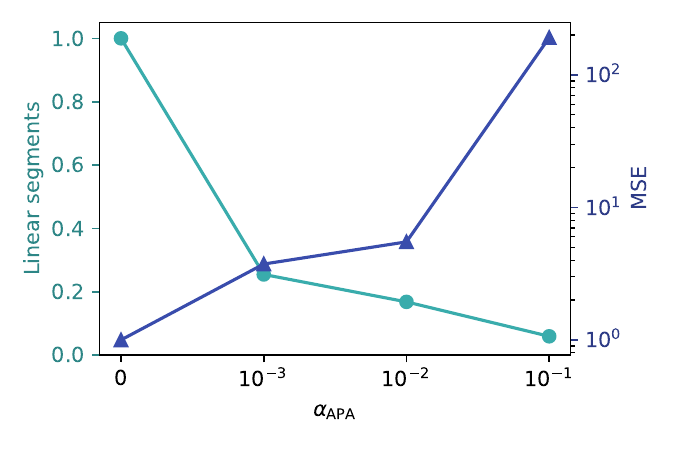}
        \caption{Dynamics networks.}
        \label{fig: APA coefficient dynamics}
    \end{subfigure}
    \begin{subfigure}[b]{0.49\textwidth}
        \centering
        \includegraphics[width=\textwidth]{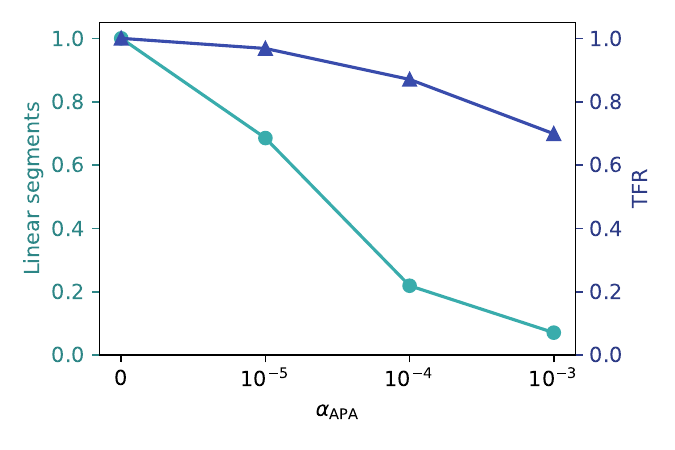}
        \caption{Value networks.}
        \label{fig: APA coefficient value}
    \end{subfigure}
    \caption{Number of linear segments and performance of dynamics networks and value networks trained under different APA coefficients. The results are averaged on the same tasks as in Figure \ref{fig: linear segments}. All scores are normalized by dividing those of $\alpha_\mathrm{APA}$=0.}
    \label{fig: APA coefficient}
\end{figure}

We compare the verification time of different regularization methods in Table \ref{tab: verify time}.
It shows that APA has the shortest verification time overall, especially in high-dimensional tasks.
In some tasks, such as Unicycle, Lane Keep, and Quadrotor, APA reduces the verification times by more than 100 times compared with no regularization.
Note that without regularization, verification may take much longer than the time limit (2 hours) in high-dimensional tasks.
The acceleration of verification brought about by APA greatly improves the scalability of our synthesis framework, enabling it to solve higher-dimensional tasks.
The superiority of APA is due to its effectiveness in reducing linear segments of both dynamics and value networks.
WD also significantly reduces the verification time compared with no regularization because it reduces linear segments of value networks.
However, the acceleration of WD is not so significant as that of APA in most nonlinear tasks because it cannot reduce linear segments of dynamics networks.
In contrast, SNR has longer verification times than no regularization, which is consistent with the fact that it increases linear segments of both dynamics networks and value networks.
SNR performs poorly because it is designed to increase the robustness of a neural network under disturbances, and the results show that this robustness-oriented objective does not always align with the objective of reducing the number of linear segments.

\begin{table}[]
    \centering
    \begin{tabular}{ccccc}
    \toprule
        Task & w/o reg & WD & SNR & APA \\
    \midrule
        Double Integrator & 2.1 & \textbf{0.2} & 0.5 & 0.5 \\
        Pendulum & 30.0 & \textbf{0.91} & 13.58 & 3.9 \\
        Unicycle & 2597.5 & 12.42 & 393.3 & \textbf{11.6} \\
        Lane Keep & 202.4 & 8.0 & 1540.9 & \textbf{1.4} \\
        Quadrotor & Timeout & 141.01 & Timeout & \textbf{2.9} \\
        Cart Pole & Timeout & \textbf{133.18} & 1009.1 & 171.5 \\
        Point Mass & 1216.6 & 61.2 & Timeout & \textbf{48.2} \\
        Robot Arm & 178.3 & 7.0 & 1071.8 & \textbf{5.5} \\
        Robot Dog & Timeout & 1043.9 & MaxIter & \textbf{311.9} \\
    \bottomrule
    \end{tabular}
    \caption{Verification time (in seconds) of different regularization methods. The dynamics networks (for nonlinear tasks) and value networks in each column are trained with the same regularization method. ``Timeout" means that fine-tuning exceeds the time limit (2 hours). ``MaxIter" means that fine-tuning exceeds the iteration limit (100k).}
    \label{tab: verify time}
\end{table}

\subsubsection{Counterexample search}
\label{sec: counterexample search}
We compare our counterexample search method, BGB, with two existing search methods: projected boundary search (PBS) proposed by \shortciteA{liu2023safe} and PGD with standard backtracking (PGD-B), to study its search efficiency.
We count the number of fine-tuning iterations and fine-tuning times of the three search methods, as shown in Table \ref{tab: search method comparison}.
Results show that BGB has the smallest number of fine-tuning iterations and fine-tuning time among the three methods, indicating that it has the highest counterexample search efficiency.
PBS has the lowest search efficiency, exceeding the iteration limit on most high-dimensional tasks.
This is because projecting the state to the boundary of feasible region in every step is unnecessary and significantly harms search efficiency\footnote{Projecting the state to the boundary of feasible region in every step is unnecessary not only for HJ reachability but also for other safety certificates such as CBF and CLF, at least in discrete-time systems, because the feasible region conditions are the same for all safety certificates.
The projection may become necessary in continuous-time systems where counterexamples must be exactly on the boundary.}.
We need the state to be close to the boundary only at the final step, not at all intermediate steps.
PGD-B also has lower search efficiency than BGB because standard backtracking can only search toward but not along the boundary, making it easy to get stuck near the boundary.

\begin{table}[]
    \centering
    \begin{tabular}{ccccccc}
    \toprule
        \multirow{2}{*}{Task} & \multicolumn{3}{c}{FT iter (k)} & \multicolumn{3}{c}{FT time (s)} \\
        \cmidrule(lr){2-4} \cmidrule(lr){5-7}
        & PBS & PGD-B & BGB & PBS & PGD-B & BGB \\
    \midrule
        Double Integrator & 1.6 & \textbf{1.2} & \textbf{1.2} & \textbf{25.3} & 50.1 & 28.5 \\
        Pendulum & 1.5 & \textbf{1.1} & \textbf{1.1} & \textbf{31.3} & 55.4 & 31.4 \\
        Unicycle & MaxIter & 54.1 & \textbf{6.9} & MaxIter & 3963.6 & \textbf{175.0} \\
        Lane Keep & 41.0 & 33.1 & \textbf{5.9} & 916.2 & 2370.0 & \textbf{184.7} \\
        Quadrotor & 21.8 & 3.1 & \textbf{1.6} & 578.3 & 223.7 & \textbf{55.3} \\
        Cart Pole & 43.8 & 62.6 & \textbf{2.4} & 1276.8 & 4590.0 & \textbf{83.3} \\
        Point Mass & MaxIter & 35.1 & \textbf{8.6} & MaxIter & 3199.2 & \textbf{321.7} \\
        Robot Arm & MaxIter & MaxIter & \textbf{25.2} & MaxIter & MaxIter & \textbf{853.8} \\
        Robot Dog & MaxIter & 51.5 & \textbf{6.3} & MaxIter & 3436.2 & \textbf{333.7} \\
    \bottomrule
    \end{tabular}
    \caption{Number of fine-tuning iterations and fine-tuning time of different counterexample search methods.}
    \label{tab: search method comparison}
\end{table}

\subsubsection{Feasible region regularization}
\label{sec: feasible region regularization}
We compare our feasible region regularization method, ESR, with RSR \shortcite{chang2019neural,liu2023safe} and no regularization to study its effectiveness in enlarging feasible regions.
Any feasible region regularization method will make fine-tuning harder because it inevitably includes some infeasible states into the zero-sublevel set.
To evaluate the negative impact on fine-tuning, we not only compute the TFR of the value networks but also count the number of fine-tuning iterations of each regularization method, as shown in Table \ref{tab: ESR TFR}.
It shows that ESR has the highest TFR on almost all tasks, significantly increasing TFR compared with no regularization, especially on high-dimensional tasks.
TFR of no regularization becomes smaller as the state dimension increases, indicating that fine-tuning tends to mistakenly exclude feasible states from the zero-sublevel set, resulting in feasible region shrinkage.
RSR also increases TFR compared with no regularization, but it is not so effective as ESR, and its number of fine-tuning iterations is not less than ESR.
This is because RSR randomly pushes all states into the zero-sublevel set, which will mistakenly include more infeasible states than ESR, resulting in lower regularization efficiency and a greater negative impact on fine-tuning.

\begin{table}[]
    \centering
    \begin{tabular}{ccccccc}
    \toprule
        \multirow{2}{*}{Task} & \multicolumn{3}{c}{TFR} & \multicolumn{3}{c}{FT iter (k)} \\
        \cmidrule(lr){2-4} \cmidrule(lr){5-7}
        & w/o reg & RSR & ESR & w/o reg & RSR & ESR \\
    \midrule
        Double Integrator & 0.897 & \textbf{0.960} & 0.940 & \textbf{1.1} & 1.2 & 1.2 \\
        Pendulum & 0.856 & 0.952 & \textbf{0.962} & \textbf{1.0} & \textbf{1.0} & \textbf{1.0} \\
        Unicycle & 0.671 & 0.907 & \textbf{0.911} & \textbf{1.2} & 3.9 & 6.9 \\
        Lane Keep & 0.651 & 0.693 & \textbf{0.750} & 3.5 & \textbf{1.6} & 5.9 \\
        Quadrotor & 0.872 & 0.901 & \textbf{0.906} & 2.2 & 2.3 & \textbf{1.6} \\
        Cart Pole & 0.014 & 0.207 & \textbf{0.404} & 4.1 & 14.5 & \textbf{2.4} \\
        Point Mass & 0.180 & 0.549 & \textbf{0.594} & \textbf{1.6} & 7.5 & 8.6 \\
        Robot Arm & 0.000 & \textbf{0.405} & 0.403 & \textbf{3.5} & 79.3 & 25.2 \\
        Robot Dog & 0.659 & 0.865 & \textbf{0.872} & \textbf{4.0} & 12.2 & 6.3 \\
    \bottomrule
    \end{tabular}
    \caption{TFR and number of fine-tuning iterations of different feasible region regularization methods.}
    \label{tab: ESR TFR}
\end{table}

\subsubsection{Ablation study}
We perform ablation studies to show how the proposed three techniques contribute to the reduction of overall synthesis time and the increase of TFR, and the results are shown in Figure \ref{fig: ablation}.

First, we test a baseline algorithm called Vanilla that directly minimizes MSE without neural network regularization in pre-training, uses PGD-B to search counterexamples, and performs fine-tuning without feasible region regularization.
Results show that Vanilla fails to synthesize value functions on three higher-dimensional nonlinear tasks, i.e., Cart Pole, Point Mass, and Robot Dog.
Moreover, it also fails on Robot Arm because the TFR is zero, i.e., the zero-sublevel set of the value function shrinks to an empty set.

\begin{figure}
    \centering
    \begin{subfigure}[b]{\textwidth}
        \centering
        \includegraphics[width=\textwidth]{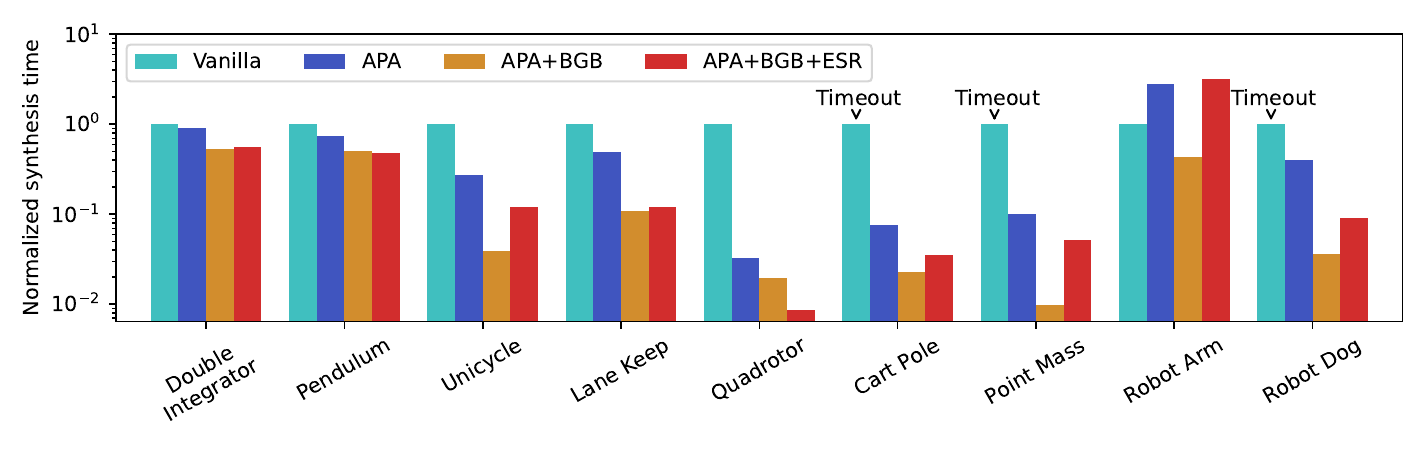}
    \end{subfigure}
    \begin{subfigure}[b]{\textwidth}
        \centering
        \includegraphics[width=\textwidth]{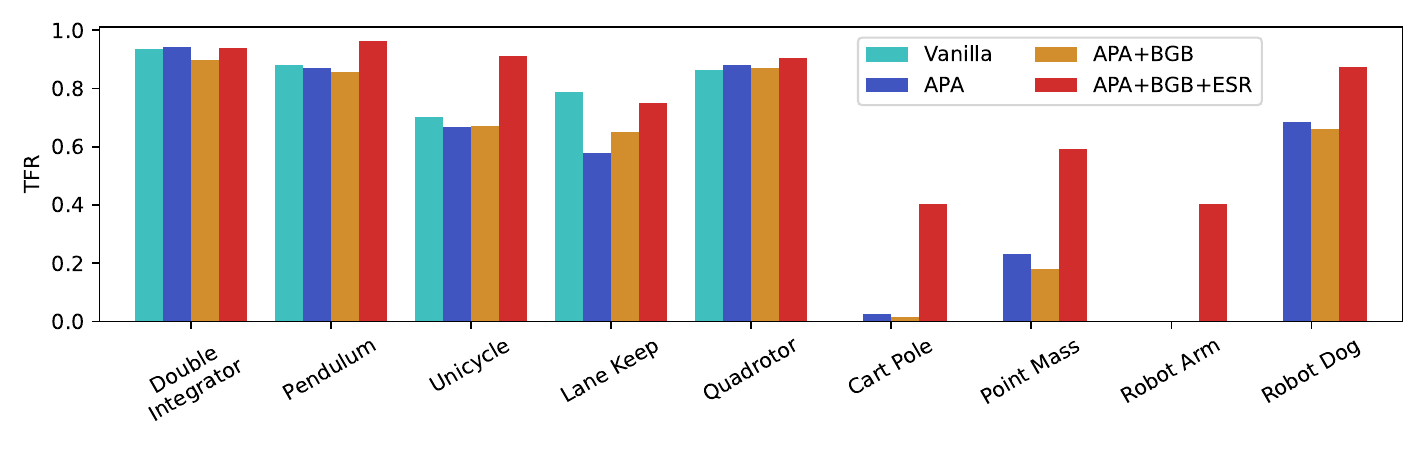}
    \end{subfigure}
    \caption{Ablation study of three techniques with respect to synthesis time and TFR. The normalized synthesis time is the synthesis time, i.e., fine-tuning time plus verification time, of each algorithm divided by that of Vanilla. The ``Timeout" annotation on top of the bars means the corresponding experiments exceed the time limit (2 hours), and we use the time limit for normalization in these tasks.}
    \label{fig: ablation}
\end{figure}

Next, we add APA in pre-training and keep the adversarial training part unchanged.
Results show that APA significantly reduces synthesis time on almost all tasks, especially higher-dimensional ones.
The comparison of synthesis time on Robot Arm is meaningless because all algorithms fail to synthesize a non-trivial value function except the last one that uses all three techniques.
APA's reduction of synthesis time is mainly attributed to its acceleration of verification, not only the last verification that proves hold but also intermediate failed verifications.

Then, we add BGB for counterexample search and keep the fine-tuning loss unchanged.
Results show that BGB further reduces synthesis time on all tasks and does not cause significant changes in TFR.
BGB's reduction of synthesis time is mainly attributed to its acceleration of counterexample search, which results in fewer fine-tuning iterations.

Finally, we add ESR to fine-tuning loss, obtaining the complete version of our algorithm.
Results show that ESR substantially increases TFR on almost all tasks, especially Cart Pole and Robot Arm, where other algorithms fail or almost fail to synthesize non-trivial value functions.
Although the synthesis times of the complete algorithm increase compared with APA+BGB on some tasks, it still achieves a large acceleration compared with Vanilla.
Except for the first two lower-dimensional tasks, the acceleration compared with Vanilla is close to or more than 10 times on all tasks, and the acceleration on Quadrotor reaches about 100 times.
Note that Vanilla exceeds the time limit on three tasks, where the acceleration of our techniques could be much greater than that shown in the figure.
These results indicate that APA and BGB significantly reduce synthesis time, ESR substantially increases TFR, and these three techniques together significantly improve the scalability of our framework.

\section{Conclusion}
This paper proposes a scalable framework for formally synthesizing verified neural HJ reachability value functions.
The framework consists of three stages: pre-training, adversarial training, and verification-guided training.
We propose three techniques that significantly improve the scalability of our framework: boundary-guided backtracking (BGB) to accelerate counterexample search, entering state regularization (ESR) to enlarge feasible regions, and activation pattern alignment (APA) to accelerate MILP-based verification.
We also provide a neural safety certificate synthesis and verification benchmark called Cersyve-9, including nine commonly used safe control tasks.
Our framework successfully synthesizes verified neural value functions on all tasks in our benchmark.
Extensive experiments show that the three proposed techniques exhibit superior scalability and efficiency compared with existing methods.
While our experiments mainly focus on the synthesis side, the proposed benchmark could also foster additional study in the verification community on scaling up verification algorithms with respect to these unique types of problems.
For future work, there is still room for further scaling this framework to higher-dimensional problems.
Also, considering the approximation error of neural network dynamic models in verification is necessary for applying this framework to real-world tasks.
Moreover, this framework can be used to synthesize and verify other neural safety certificates, such as neural CBFs and CLFs.

\acks{This work is in part supported by the National Science Foundation under Grant No. 2144489. Any opinions, findings, and conclusions or recommendations expressed in this material are those of the authors and do not necessarily reflect the views of the National Science Foundation.}



\vskip 0.2in
\bibliography{sample}
\bibliographystyle{theapa}

\end{document}